\RequirePackage{fix-cm}
\documentclass[smallcondensed]{svjour3}     
\renewcommand{\makeheadbox}[1]{}
\smartqed  

\usepackage{url,alltt}
\usepackage{paralist}
\usepackage{xspace}
\usepackage{caption}
\usepackage{subcaption}
\usepackage{graphicx}
\usepackage{wrapfig}
\usepackage{times}
\usepackage{amsmath}
\usepackage{amssymb}
\usepackage{pifont}
\usepackage{footnote}
\usepackage{latexsym}
\usepackage{pdfpages}
\usepackage{array}
\usepackage{verbatim}
\usepackage{stmaryrd}
\usepackage{ifsym}
\usepackage{algorithmic}
\usepackage{algorithm}
\usepackage{gastex} 
\usepackage{parcolumns}
\usepackage{listings}
\usepackage{color}
\usepackage{hyperref}
\usepackage{enumitem}
\usepackage{pdfpages}
\usepackage{mathptmx}      

\definecolor{dkgreen}{rgb}{0,0.6,0}
\definecolor{gray}{rgb}{0.5,0.5,0.5}
\definecolor{mauve}{rgb}{0.58,0,0.82}

\setlist{nolistsep}

\newcommand{\Nat}{\ensuremath{\mathbb{N}}}

\newcommand{\false}{\mbox{False}}
\newcommand{\true}{\mbox{True}}

\newcommand{\simplies}{\DOTSB\Longrightarrow}

\newcommand{\exref}[1]{Example~\ref{examp:#1}}

\newcommand{\secref}[1]{Section~\ref{sec:#1}}
\newcommand{\lemref}[1]{Lemma~\ref{lem:#1}}

\newcommand{\figref}[1]{Figure~\ref{Fi:#1}}
\newcommand{\thmref}[1]{Theorem~\ref{thm:#1}}
\newcommand{\gV}[1]{\langle \textit{#1}\rangle}
\newcommand{\bid}{\textbf{id}}
\newcommand{\brcv}{\textbf{input}}

\newcommand{\bfrw}{\textbf{output}}
\newcommand{\bforward}{\bfrw}

\newcommand{\bif}{\textbf{if}}
\newcommand{\bthen}{\textbf{then}}

\newcommand{\binsert}{\textbf{insert}}
\newcommand{\bremove}{\textbf{remove}}
\newcommand{\babort}{\textbf{abort}}

\newcommand{\bin}{\textbf{in}}
\newcommand{\band}{\textbf{and}}
\newcommand{\bor}{\textbf{or}}
\newcommand{\bnot}{\textbf{not}}

\newcommand{\btuple}[1]{\overline{#1}}
\newcommand{\tconnected}{\texttt{connected}}
\newcommand{\tallPorts}{\texttt{allPorts}}
\newcommand{\ttrusted}{\texttt{trusted}}
\newcommand{\tcache}{\texttt{cache}}

\newcommand{\tforbidden}{\texttt{forbidden}}

\newcommand{\N}{\mathbb{N}}

\newcommand{\tnextport}{\texttt{nextport}}

\newcommand{\bwhile}{\textbf{while}}
\newcommand{\bforeach}{\textbf{foreach}}
\newcommand{\breturn}{\textbf{return}}

\newcommand{\Z}{\mathbb{Z}}

\newcommand{\Heading}[1]{\vspace{-0.3cm}\paragraph{{#1}}}
\newcommand{\BeginProof}{\vspace{-0.25cm}\begin{proof}}

\newcommand{\Comment}[1]{}
\newcommand{\Appendix}[1]{}

\newcommand{\AppendixContent}[1]{}

\newcommand{\ie}{{\it i.e.,}\xspace}

\newcommand\eat[1]{}

\newcommand{\Net}{\mathsf{N}}
\newcommand{\PortSet}{\mathsf{Pr}}
\newcommand{\SinglePort}{\mathit{pr}}
\newcommand{\ForwardFunc}{\mathsf{f}}
\newcommand{\ForwardRel}{\mathsf{f_r}}

\newcommand{\allnotes}[1]{}
\renewcommand{\allnotes}[1]{\textit{#1}}

\newcommand{\AbsPackets}{{P}}

\begin{document}

\title{Some Complexity Results for Stateful Network Verification
  \thanks{A preliminary version of this work appeared in \cite{velner2016some}.
  }
}

\author{Kalev Alpernas \and
Aurojit Panda \and
Alexander Rabinovich \and
Mooly Sagiv \and
Scott Shenker \and
Sharon Shoham \and
Yaron Velner
}

\institute{
	K. Alpernas \at
		Tel Aviv University\\
		\email{kalev.alp@gmail.com}
	\and
	A. Panda \at
		NYU
	\and
	A. Rabinovich \at
		Tel Aviv University
	\and
	M. Sagiv \at
		Tel Aviv University
	\and
	S. Shenker \at
		UC Berkeley
	\and
	S. Shoham  \at
		Tel Aviv University
	\and
	Y. Velner \at
		Hebrew University of Jerusalem
}

\date{}

\maketitle

\begin{abstract}
In modern networks, forwarding of packets often depends on the history of
previously transmitted traffic. Such networks contain \emph{stateful}
middleboxes, whose forwarding behaviour depends on a mutable internal state.
Firewalls and load balancers are typical examples of stateful middleboxes.

This work addresses the complexity of verifying safety properties, such as
isolation, in networks with finite-state middleboxes. Unfortunately, we show
that even in the absence of forwarding loops, reasoning about such networks is
undecidable due to interactions between middleboxes connected by unbounded
ordered channels. We therefore abstract away channel ordering. This abstraction
is sound for safety, and makes the problem decidable. Specifically, safety
checking becomes EXPSPACE-complete in the number of hosts and middleboxes in the
network.
To tackle the high complexity, we identify two useful subclasses of finite-state
middleboxes which admit better complexities. The simplest class includes, e.g.,
firewalls and permits polynomial-time verification.
The second class includes, e.g., cache servers and learning switches, and makes
the safety problem coNP-complete.

Finally, we implement a tool for verifying the correctness of stateful networks.

\end{abstract}

\section{Introduction}
\label{sec:Intro}

Modern computer networks are extremely complex, leading to many bugs and vulnerabilities which affect our daily life.
Therefore, network verification is an increasingly important topic addressed by the programming languages and networking communities
(e.g., see \cite{conext:uzniarPCVK12,nsdi:CaniniVPKR12,nsdi:KVM12,CCR:KhurshidZCG12,nsdi:KCZCMW13,FMACAD:SNM13,FlowLog,NetKat15}).
Previous network verification tools leverage a simple network forwarding model which renders the datapath {\em immutable};
\ie normal packets going through the network do not change its forwarding behaviour, and the control plane explicitly alters the forwarding state at relatively slow time scales. Thus, invariants can be verified before each control-plane initiated change and these invariants will be enforced until the next such change. While the notion of an immutable datapath supported by an assemblage of routers makes verification tractable, it does not reflect reality. Modern enterprise networks are comprised of roughly $2/3$ routers
\footnote{In this work we do not distinguish between routers and switches,
since they obey similar forwarding models.}
and $1/3$ {\em middleboxes}~\cite{sherry2012making}.
A simple example of a middlebox is a stateful hole-punching firewall which permits traffic from untrusted hosts only after they have received a message from a trusted host. Middleboxes --- such as firewalls, WAN optimizers, transcoders, proxies, load-balancers, intrusion detection systems (IDS) and the like --- are the most common way to insert new functionality in the network datapath, and are commonly used to improve network performance and security. While useful, middleboxes are a common source of errors in the network~\cite{potharaju2013demystifying}, with middleboxes being responsible for over $40\%$ of all major incidents in networks.

This work addresses the problem of verifying safety of networks with middleboxes, referred to as \emph{stateful} networks.
We model such a network as a finite undirected graph with two types of nodes:
(i)~hosts which can send packets,
(ii)~middleboxes which react to packet arrivals and forward modified packets.
Each node in the network has a fixed number of ports, connected by network edges (links).

From a verification perspective, it is possible to view a middlebox as a procedure with local mutable state which is atomically changed every time a packet is transmitted. The local state determines the forwarding behaviour.\footnote{Routers may be considered a degenerate case of middleboxes, whose state is constant and hence their forwarding behaviour does not change over time.} Thus, the problem of network verification amounts to verifying the correctness of a specialized distributed system where each of the middleboxes operates atomically and the order of packet processing by different middleboxes is arbitrary.

Real middleboxes are generally complex software programs implemented in several hundreds of thousands of lines of code.
We follow \cite{panda2014verifying,SNAPL:PandaASSS15} in assuming that we are provided with middlebox models in the form of \emph{finite-state transducers}. In our experience one can naturally model the behaviour of most middleboxes this way.
For every incoming packet, the transducer uses the packet header and the local state to compute the forwarding behaviour (output) and to
update its state for future packets.
The  transducer can be non-deterministic to allow modelling of middleboxes like load-balancers whose behaviour depends not just on the state, but also on a random number source.
We symbolically represent the local state of each middlebox by a fixed set of relations on finite elements, each with a fixed arity.

\paragraph{The Verification Problem}
We define network safety by means of avoiding ``bad'' middlebox states (e.g., states from which a middlebox forwards a packet in a way that violates a network policy).
Given a set of bad middlebox states, we are interested in showing that for all packet scenarios the bad states cannot be reached.
This problem is hard since the number of packets is unbounded
and the states of one middlebox can affect another via transmitted packets.

\subsection{What is Decidable About Middlebox Verification}
\label{sec:Decidability}
In \secref{undecidability}, we prove that for general stateful networks the
verification problem is undecidable. This result relies on the observation that
packet histories can be used to count, similarly to results in model checking of
infinite ordered communication channels~\cite{brand1983communicating}.
Simulating counting is immediate when the network configuration admits forwarding loops. However, such
configurations are usually avoided in real networks. In order to address
realistic networks, we show that the verification problem is undecidable even
for networks without forwarding loops.

In order to obtain decidability, we introduce an abstract semantics of networks where the order of packet processing on each channel (connecting two middleboxes or a middlebox and a host) is arbitrary, rather than first-in, first-out (FIFO).
Thus, middlebox inputs are multisets of packets which can be processed in any order.
This abstraction is \emph{conservative}, i.e., whenever we verify that the network does not reach a bad state, it is indeed the case.
However, the verification may fail even in correct networks, resulting in false alarms.
Since packets are atomically processed, we note that network designers can impose ordering even in this abstract model
by sending acknowledgments for received packets, and dropping out-of-order packets.

In fact, the abstraction of the packet order over channels closely corresponds to assumptions made by network engineers: since packets in modern networks can traverse multiple paths,
be buffered, or be chosen for more complex analysis, network software cannot assume that packets sent from a source to a server are received by a server in order.
Network protocols therefore commonly build on TCP,
a protocol which uses acknowledgments and other mechanisms to ensure that servers receive packets in order.
Since packet ordering is enforced by causality (by sending acknowledgments) and by software on the receiving end,
rather than by the network semantics, correctness of such networks typically does not rely on the order of packet processing.
Therefore we can successfully verify a majority of network applications despite our abstraction. 

\subsection{Complexity of Stateful Verification}
\label{sec:Complexity}
In \secref{complexityLowerBounds}, we show that the problem of network
verification when assuming a nondeterministic order of packet processing is
complete for exponential space, i.e., it is decidable, and in the worst case,
the decision procedure can take exponential space in terms of hosts and
middleboxes. This is proved by showing that the network safety problem is
equivalent to the coverability problem of Petri nets, which is known to be
EXPSPACE-complete~\cite{rackoff1978covering,lipton1976reachability}. This result
is not surprising, and resembles previous work on message passing systems with
unordered communication channels~\cite{lipton1976reachability,sen2006model}.

\begin{figure}
\centering
\includegraphics[width=0.7\textwidth]{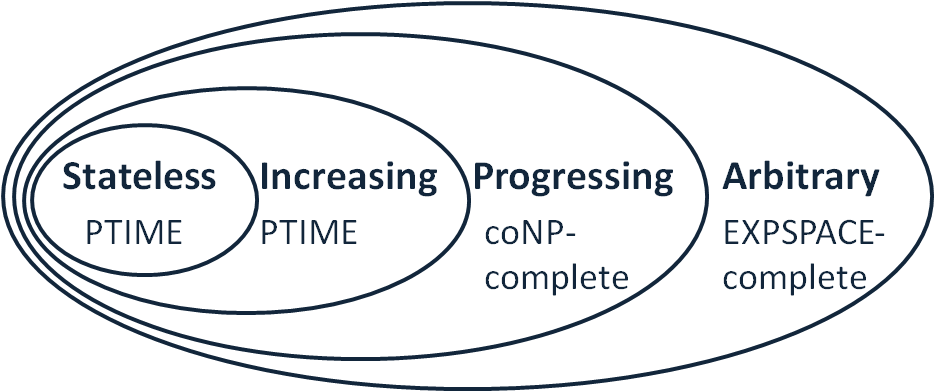}
\caption{Middlebox hierarchy with worst-case time complexity for each category.}
\label{Fi:Complexity}
\end{figure}

Since the problem is complete, it is impossible to improve this upper-bound without further assumptions.
Therefore, we consider limited cases of middleboxes permitting more efficient verification procedures, as shown in \figref{Complexity}.
We identify four classes of middleboxes with increasing expressive power and verification complexity:
(i)~\emph{stateless} middleboxes whose forwarding behaviour is constant over time,
(ii)~\emph{increasing} middleboxes whose forwarding behaviour increases over time, i.e., as the history of packets is extended, the set of forwarded packets may never decrease,
(iii)~\emph{progressing} middleboxes whose forwarding behaviour cannot regress to a previous state, i.e., the transition relation of the transducer does not include cycles besides self-cycles (or cycles of ``equivalent states''), and hence the forwarding behavior stabilizes after some 
time,
and
(iv)~\emph{arbitrary} middleboxes without any restriction.
For example, NATs, Switches and simple ACL-based firewalls are stateless; hole-punching stateful firewalls are increasing -- as time proceeds more hosts become ``trusted'' and hence more packets are being forwarded rather than dropped; and learning-switches and cache-proxies are progressing and not increasing -- information that is learnt is never unlearned, but the forwarding behaviour may decrease as a result of learning (e.g., a learning switch would forward a packet to the right port rather than broadcasting it).

For stateless and increasing middleboxes, we prove that any packet which arrives once can arrive any number of times, leading to a polynomial-time verification algorithm, using a fixed-point computation.
We note that efficient near linear-time algorithms for stateless verification are known (e.g., see~\cite{CCR:KhurshidZCG12}).
Our result generalizes these results to increasing networks and is in line with the recent work in~\cite{nsdi/FogelFPWGMM15,nsdi/LopesBGJV15}.

For progressing middleboxes, we show that verification is coNP-complete.
The main insight is that if a bad state is reachable then there exists a small (polynomial) input scenario leading
to a bad state.
This means that tools like SAT solvers which are frequently used for verification can be used to verify large networks in many cases
but it also means that we cannot hope for a general efficient solution unless P=NP.

Finally, we note that unlike the known results in stateless networks, the absence of forwarding loops does not improve the upper bound, i.e.,
we show that our lower bounds also hold for networks without forwarding loops.

\paragraph{Packet Space Assumption}
Previous works in stateless verification~\cite{nsdi:KVM12,NetKat15} assume that packet headers have $n$-bits, simulating realistic
packet headers which can be large in practice.
This makes the complexity of checking safety of stateless networks PSPACE-hard.
Our model avoids packet space explosion by only supporting three fields: source, destination, and packet tags. We make this simplification since our work primarily focuses on middlebox policies (rather than routing). 
As demonstrated in \secref{Examples}, middlebox policies are commonly specified in terms of the source and destination hosts of a packet and the network port (service) being accessed. For example, at the application level, firewalls may decide how to handle a packet according to a small set of application types (e.g., skype, ssh, etc.).
Source, destination and packet tag are thus sufficient for reasoning about safety with respect to these policies. This simplification is also supported by recent works (e.g.~\cite{CCR:KhurshidZCG12}) which suggest that in practice the forwarding behaviour depends only on a small set of bits.

\paragraph{Lossless Channels}
Previous works on infinite ordered communication channels have introduced \emph{lossy channel systems}~\cite{abdulla1993verifying} as an abstraction of ordered communication that recovers decidability. Lossy channel systems allow messages to be lost in transit, making the reachability problem decidable, but with a non-elementary lower bound on time complexity. In our model, packets cannot be lost. On the other hand, the order of packets arrival becomes nondeterministic. With this abstraction, we manage to obtain elementary time complexity for verification.

\paragraph{Initial Experience}

We implemented a tool which accepts symbolic representations of middleboxes and
a network configuration and verifies safety. For increasing (and stateless) networks, the tool generates a Datalog program and a query which holds iff a bad state is reachable.
Then, the query is evaluated using existing Datalog engines~\cite{LogicBlox}.

For arbitrary networks (and for progressing networks), the tool generates a petri-net and a coverability property which holds iff the
network reaches a bad state.
To verify the coverability property we use LOLA~\cite{schmidt2000lola,TRLola2} --- a Petri-Net model checker.

\subsection{Main Results and Outline}
This work addresses the complexity of verifying the safety of stateful networks. It makes the following main contributions:

\begin{itemize}
\item We introduce a formal model for stateful networks with finite-state
middleboxes, inspired by communicating finite state
machines~\cite{brand1983communicating} (\secref{FormalModel}).
We further propose a symbolic representation of middleboxes, resulting in some cases in an exponentially more succinct
description compared to an explicit representation as a finite state machine
(\secref{mboxes}).
We use the formal model to show that verifying safety properties in stateful networks is undecidable,
even when the network configuration does not admit forwarding loops
(\secref{undecidability}).
\item We adopt an unordered abstraction inspired by
\cite{lipton1976reachability,sen2006model} to define a conservative abstraction
of networks in which packets can be processed out of order (\secref{unordered}). Under this
abstraction, the safety problem of stateful networks becomes decidable, but
EXPSPACE-complete. Interestingly, we show that for a certain class of networks (namely, increasing networks) this abstraction is in fact precise for safety (\secref{MboxClass}).
\item We identify four classes to which we classify networks, characterized by the forwarding
behaviours of their middleboxes: stateless, increasing, progressing and arbitrary (\secref{MboxClass}).
We characterize these classes both semantically and syntactically via restrictions on the symbolic representation of the middleboxes (\secref{classes-symbolic}),
and demonstrate that these classes capture real-world middleboxes (\secref{Examples}).
\item We show that different network classes admit better complexity results than the EXPSPACE complexity of arbitrary networks: PTIME for stateless and increasing networks (\secref{upper-increasing}), and coNP for progressing networks (\secref{AcyclicInNP}).
The upper bounds are
made more realistic by stating them in terms of a symbolic representation of
middleboxes, i.e., the middlebox code, rather than the explicit state space.
We match the upper bounds with lower bounds
(\secref{complexityLowerBounds}), which are obtained with standard middleboxes,
and thus reflect the complexity of realistic networks.

\item We present initial empirical results using Petri nets and Datalog engines
to verify safety of networks (\secref{benchmark}).
\end{itemize}
Finally, we discuss related work and conclude in \secref{Related}.
\section{A Formal Model for Stateful Networks}\label{sec:FormalModel}
In this section, we present a formal model of networks with
stateful middleboxes. We define a concrete network semantics,
and present the \emph{safety} verification problem, as well as the special case of \emph{isolation}.
Finally, we show that the safety verification problem is undecidable under the concrete semantics.

A \emph{network} $\Net$ is a finite undirected graph of \emph{hosts} and \emph{middleboxes}, equipped with a \emph{packet domain}.
Formally, $\Net = (H\cup M, E, P)$, where $H$ is a finite set of \emph{hosts}, $M$ is a finite set of \emph{middleboxes}, $E \subseteq \{\{u,v\}\mid u,v\in H\cup M\}$ is the set of (undirected) edges and $P$ is a set of packets.

\paragraph{Packets}
In real networks, a packet
consists of a \emph{packet header} and a \emph{payload}.
The packet header contains a source and a destination host ids and additional arbitrary stream of control bits.
The payload is the content of the packet and may consist of any arbitrary sequence of bits.
The cardinality of the set of packets is determined by the possible range of control bits and the possible space of payloads, and need not be finite.

In this work, $P$ is a set of \emph{abstract packets}.
An abstract packet $p \in P$ consists of a header only, in the form of a triple $(s,d,t)$, where $s,d\in H$ are the source and destination hosts (respectively) and $t$ is a \emph{packet tag} that ranges over a finite domain $T$.
Intuitively, $T$ stands for an abstract set of services or security policies.
Therefore, $\AbsPackets = H\times H\times T$ is a finite set.

Middlebox behaviour in our model is defined with respect to abstract packets and is oblivious of the underlying concrete packets.

Each host $h \in H$ is associated with a set of packets that it can send, denoted $P_h \subseteq P$.

\subsection{Stateful Middleboxes} \label{sec:mboxes}

A \emph{middlebox} $m \in M$ in a network $\Net$ has a set of \emph{ports} $\PortSet$ and a \emph{forwarding transducer} $F$. The set of \emph{ports} $\PortSet$ consists of all the adjacent edges of $m$ in the network $\Net$,

The forwarding transducer of a middlebox is a tuple $F = (\Sigma, \Gamma, Q_m, q^0_m, \delta_m)$ where:

\begin{itemize}
    \item $\Sigma = P \times \PortSet$ is the input alphabet in which each letter consists of a packet and an input port,
    \item $\Gamma = 2^{P \times \PortSet}$ is the output alphabet in which each letter describes (possibly empty) sets of packets sent over the different ports,
 	\item $Q_m$ is a possibly infinite set of states,
 	\item $q^0_m \in Q_m$ is the initial state, and
 	\item $\delta_m \subseteq Q_m \times \Sigma \times {\Gamma \times Q_m}$ is the transition relation, which describes both the output and the change of state in response to an input.
\end{itemize}

Note that the alphabet $\Sigma$ is finite (since abstract packets are considered).

We often refer to the transition relation $\delta_m$ as a function $\delta_m: Q_m \times \Sigma \to  2^{\Gamma \times Q_m}$,
where $\delta_m(q,(p,\SinglePort)) = \{(o, q') \mid (q, (p,\SinglePort), o, q') \in \delta_m \}$.
If $\delta_m(q, (p,\SinglePort)) = \emptyset$, we say that $\delta_m$ is undefined for the packet $p$ arriving on port $\SinglePort$ in state $q$.

We extend $\delta_m$ to sequences $h \in (P\times \PortSet)^*$  in the natural way: $\delta_m(q,\epsilon) = \{(\epsilon, q)\}$ and $\delta_m(q,h\cdot (p,\SinglePort)) = \{ (\gamma_i \cdot o', q') \mid \exists q_i \in Q_m.\,  (\gamma_i,q_i) \in \delta_m(q,h) \wedge (o',q') \in \delta_m(q_i,(p,\SinglePort))\}$.
The language of a state $q \in Q_m$ is $L(q) =  \{(h,\gamma) \in  (P\times \PortSet)^* \times (2^{P\times \PortSet})^* \mid \exists q' \in Q_m.\ (\gamma,q') \in \delta_m(q,h)\}$.
The language of $F$, denoted $L(F)$, is the language of $q^0_m$.
We also define the set of \emph{histories} leading to $q \in Q_m$ as $h(q) = \{h \in (P\times \PortSet)^* \mid \exists \gamma \in (2^{P\times \PortSet})^* .\ (\gamma,q) \in \delta_m(q^0_m,h) \}$.

$F$ is deterministic if for every $q \in Q_m$ and every $(p, \SinglePort) \in \Sigma$, $|\delta_m(q,(p,\SinglePort))|\leq 1$. If $F$ is deterministic, then every history leads to at most one state and output, in which case $F$ defines a (possibly partial) \emph{forwarding function} $\ForwardFunc : (P\times \PortSet)^* \times (P \times \PortSet) \to 2^{P\times \PortSet}$ where
$\ForwardFunc(h, (p,\SinglePort)) = o$ for the  (unique) output $o \in 2^{P\times \PortSet}$ such that $(h \cdot (p,\SinglePort), \gamma \cdot o) \in L(F)$ for some $\gamma \in (2^{P\times \PortSet})^*$.
If no such output $o$ exists, then $\ForwardFunc$ is undefined.
The forwarding function $\ForwardFunc$ defines the (possibly empty) set of output packets (paired with output ports) that $m$ will send to its neighbors
following a history $h$ of packets that $m$ received in the past and input packet $p$ arriving on input port $pr$.
We note that $\ForwardFunc(h,(p,\SinglePort))=\emptyset$ should not be confused with the case where $\ForwardFunc(h,(p,\SinglePort))$ is undefined.

If $F$ is nondeterministic, a \emph{forwarding relation} $\ForwardRel \subseteq(P\times \PortSet)^* \times (P \times \PortSet) \times 2^{P\times \PortSet}$ is defined in a similar way.

Note that every forwarding function $\ForwardFunc$ can be defined by an infinite-state deterministic transducer: $Q_m$ will include a state for every possible history, with $\epsilon$ as the initial state.
The transition relation $\delta_m$ will map a state and an input packet to the set of output packets as defined by $\ForwardFunc$, and will change the state by appending the packet to the history.

\subsubsection{Finite-State Middleboxes}
Arbitrary middlebox functionality, defined via infinite-state transducers, makes middleboxes Turing-complete, and hence impossible to analyze. To make the analysis tractable, we focus on abstract middleboxes, whose forwarding behaviour is defined by \emph{finite-state} transducers. Nondeterminsm can then be used to overapproximate the behaviour of a concrete, possibly infinite-state, middlebox via a finite-state abstract middlebox, allowing a sound abstraction w.r.t. safety.

In the sequel, unless explicitly stated otherwise, we consider finite-state middleboxes. We identify a middlebox with its forwarding relation and the transducer that implements it, and use $m$ to denote each of them.

\subsubsection{Symbolic Representation of Middleboxes} \label{sec:symbolic-rep}

To allow a more succinct representation, we use a symbolic representation of finite-state middleboxes, where
a state of a middlebox $m$ is described by the valuation of a finite set of relations $R_1,\ldots, R_k$ defined over finite domains (e.g., hosts). The transition relation $\delta_m$ is also described symbolically using (nondeterministic) update operations of the relations and output.
The syntax for the symbolic representation is described in \figref{SymbSyntax}.

Technically, we describe $\delta_m$ on an input packet $(\mathit{src},\mathit{dst},\mathit{tag})$  arriving from port $\mathit{prt}$ by a sequence of loop-free guarded commands, which we call a \emph{guarded command block}. Each guarded command in the block consists of a command and a guard, which determines whether the command should be executed. Guards are Boolean expressions over \emph{relation membership predicates} of the form $\overline{e}$ \bin~$R$ (where $\overline{e}=(e_1,\dots,e_n)$ for an $n$-ary relation $R$) and \emph{element equalities} $e_1 = e_2$. Each $e_i$ is either a constant or a variable that refers to packet fields: \emph{src}, \emph{dst}, \emph{tag}, \emph{prt}. Commands are of the form:
\begin{enumerate} [label=(\emph{\roman*})]
\item $\bfrw$ set of tuples,
\item $\babort$,
\item $\binsert$ tuple $\overline{e}$ to relation $R$,
\item $\bremove$ tuple $\overline{e}$ from relation $R$,
\item \emph{sequential composition}, and
\item \emph{guarded command block}.
\end{enumerate}

The semantics of $\binsert$, $\bremove$ and sequential composition is straightforward. An $\bfrw$ command produces output. In case more than one $\bfrw$ is executed, e.g., in the case of a sequential composition of $\bfrw$ commands, the output of the execution is the union of all $\bfrw$ commands. Blocks of guarded commands are executed non deterministically. That is, all the guards in the block are evaluated, and one command whose guard is evaluated to \emph{true} is executed. If no guard evaluates to \emph{true} then the empty set is produced as output, and no relation changes are made.
The $\babort$ command is used to signify that $\delta_m$ is not defined on the given input.

A symbolic middlebox program represents a finite-state middlebox where each state represents an interpretation (state) of all the relations, and the transition relation
is defined by the main guarded command block in the natural way. Note that since all the relations in the program are over finite domains, the set of states is indeed finite.

\begin{figure}
\[
\begin{array}{|lclr|}
\hline\hline

\gV{mbox} & ::= & \brcv(src,dst,tag,prt) : \gV{gcmd} ~ [\gV{gcmd}]^* &\\
\hline

\gV{gcmd} & ::= & \gV{grd} \Rightarrow \gV{cmd} & \mbox{guarded command}\\
\hline

\gV{cmd} & ::= & \bfrw ~\{\btuple{\gV{exp}} ~[,~\btuple{\gV{exp}}]^*\}  & \mbox{output a packet}\\
    & | & \babort & \mbox{terminate-abnormally}\\
    & | & \bid . \binsert ~\btuple{\gV{exp}} & \mbox{add tuple to relation}~\bid\\
    & | & \bid . \bremove ~\btuple{\gV{exp}} & \mbox{remove tuple from}~\bid \\
    & | & \gV{cmd} ; \gV{cmd} & \mbox{sequence of commands}\\
    & | & \gV{gcmd} ~ [\gV{gcmd}]^* & \mbox{guarded command block}\\
\hline

\gV{exp} & ::= & src~|~dst~|~tag~|~prt & \mbox{variable}\\
    & | & \textbf{constant} & \mbox{constant}\\
\hline

\gV{grd} & ::= & \gV{grd} ~\band ~ \gV{grd} &\\
        & | & \gV{grd} ~ \bor ~ \gV{grd} &\\
        & | & \bnot \gV{atom}&\\
        & | & \gV{atom}&\\
\hline

\gV{atom} & ::= & \gV{exp} =  \gV{exp} & \mbox{equality}\\
        & | & \btuple{\gV{exp}} ~ \bin ~ \bid & \mbox{membership test}\\

\hline\hline
\end{array}
\]
\caption{\label{Fi:SymbSyntax}%
A simple language for representing finite state middleboxes.
$\btuple{\gV{exp}}$ denotes a vector of $\gV{exp}$ separated by commas.
}
\end{figure}

\begin{lemma}\label{lem:SymbToFiniteState}
Every finite-state middlebox has a symbolic representation.
\end{lemma}

\begin{proof}
Let $Q = \{q_0,\dots,q_n\}$ be the finite set of states of $m$, and $q_0$ be the initial state.
We construct a symbolic middlebox program $A$ over the constants $q_0,\dots,q_n$ with a single unary relation $R$. Initially, $R = \{q_0\}$.
Each transition $(q',o) \in \delta_m (q,(p,pr))$ of $m$ is represented by a guarded command in the main guarded command block. The guard checks whether $q \in R$ and whether the packet is $(p,pr)$. The command is a sequential composition of three commands: The first command removes the (only) current state $q$ from $R$. The second inserts the new state $q'$ and the third outputs the tuples in $o$ according to $\delta_m$. If $\delta_m(q,(p,pr)) = \emptyset$, the $\babort$ command is used.
\qed
\end{proof}

\begin{remark}
We note that the construction of a symbolic representation in \lemref{SymbToFiniteState} results in a linear blowup of the representation, whereas the construction of the explict-state middlebox represented by a symbolic representation potentially results in an exponential blowup, suggesting that the symbolic representation is at least as succinct and is potentially exponentially more succinct than the explicit state representation.
\end{remark}

\begin{example} \label{Ex:firewall} \label{Ex:proxy}
\figref{Firewall} contains a symbolic representation of a hole-punching Firewall which uses a unary relation \ttrusted.
It assumes that port $1$ connects hosts inside a private organization to the firewall and that port $2$ connects public hosts.
By default, messages from public hosts are considered untrusted and are dropped.
\ttrusted\ is a unary relation which
stores public hosts that become trusted once they receive
a packet from private hosts.

\begin{figure}[t]
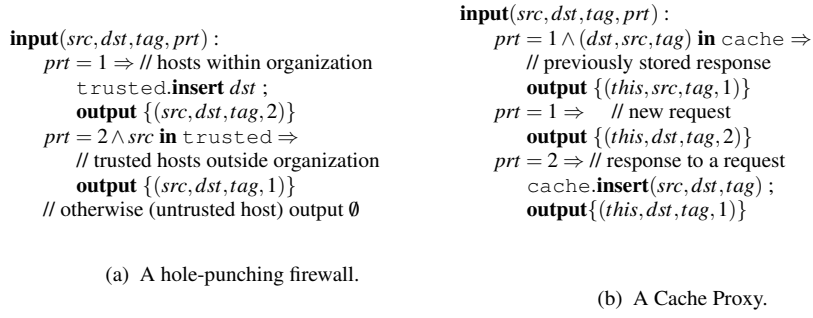

\begin{subfigure}{0.48\textwidth}
\begin{tabbing}
mm\=mm\=mm\=mm\=mm\=\kill
$\brcv(src, dst, tag, prt):$ \\
\>$ prt = 1$ $\Rightarrow$  // hosts within organization\\
\>\>$\ttrusted.\binsert~ dst$ ;\\
\>\>$\bfrw~\{(src, dst, tag,2)\}$\\
\>$prt = 2  \wedge src~\bin~\ttrusted$ $\Rightarrow$ \\
\>\> // trusted hosts outside  organization\\
\>\>$\bfrw~\{(src, dst, tag ,1)\}$\\
\> // otherwise (untrusted host) output $\emptyset$
\end{tabbing}
\caption{\label{Fi:Firewall} A hole-punching firewall.}
\end{subfigure}
\begin{subfigure}{0.48\textwidth}
\begin{tabbing}
mm\=mm\=mm\=mm\=mm\=\kill
 $\brcv(src, dst, tag, prt):$\\
\>$prt = 1 \wedge (dst, src, tag) ~\bin ~\tcache$ $\Rightarrow$ \\
\>\>// previously stored response\\
\>\>$\bfrw ~\{(this, src, tag, 1)\}$\\
\>$prt = 1$ $\Rightarrow$ \quad // new request\\
\>\>$\bfrw~\{(this, dst, tag, 2)\}$\\
\> $prt = 2 \Rightarrow$ // response to a request\\
\>\> $\tcache.\binsert (src, dst, tag)$ ;\\
\>\> $\bfrw \{(this,dst,tag, 1)\}$ \\
\end{tabbing}
\caption{\label{Fi:Proxy} A Cache Proxy.}
\end{subfigure}
\caption{Symbolic representation of middleboxes.}
\end{figure}

\figref{Proxy} contains a simplified, nondeterministic, version of a Proxy server (or cache server).
A proxy stores copies of documents (packet payloads) that passed through it. Subsequent requests for those documents are provided by the proxy, rather than being forwarded.
Technically, the middlebox has two ports, namely, a request port from which requests are received and a response port from which responses arrive.
Our modelling abstracts away the packet payloads and keeps only their types. Consequently we use nondeterminism to also account for different requests with the same type.
The internal relation $\tcache$ stores responses for packet types.
\end{example}

\subsection{Concrete (FIFO) Network Semantics}
The concrete semantics of a network $\Net = (H\cup M, E, P)$ is given by a transition system defined over a set of configurations.
In order to define the semantics we first need to define the notion of \emph{channels} which capture the transmission of packets in the network.
Formally, each (undirected) edge $\{u,v\}\in E$ in the network induces two directed \emph{channels}: $(u,v)$ and $(v,u)$. The channel $(v,u)$ is an \emph{ingress channel} of $u$, as well as an \emph{egress channel} of $v$. It consists of the sequence of packets that were sent from $v$ to $u$ and were not yet received by $u$ (and similarly for the channel $(u,v)$).
The capacity of  channels is unbounded, that is, the sequence of packets may be arbitrarily long.
In the concrete semantics, the channels admit a first-in-first-out (FIFO) behavior: Whenever a middlebox forwards a packet $p$ from a certain port it removes it from the head of the corresponding ingress channel and adds the generated packets to the tails of the corresponding egress channels (note that the mapping between channels and middlebox ports is unique).

\paragraph{Configurations and Runs}
A \emph{configuration} of a network consists of the content of each channel and the state of every middlebox.
Channels have an unbounded capacity, resulting in an infinite number of configurations even for finite state middleboxes.
The \emph{initial configuration} of a network consists of empty channels and initial states for all middleboxes.
A configuration $c_2$ is a \emph{successor} of configuration $c_1$ if it can be obtained by either: (i)~some host $h$ sending a packet $p\in P_h$ to a neighbor, thus appending the packet $p$ to the corresponding channel; or (ii)~some middlebox $m$ processing a packet $p$ from the head of one of its ingress channels, changing its state to $q'$ and appending output $o$ to its egress channels where $q',o$ are defined in accordance with $\delta_m$, i.e., if $q$ is the current state of $m$ and $\SinglePort$ is the port associated with the ingress channel then $(o,q')\in \delta_m(q,(p,\SinglePort))$.
This model corresponds to asynchronous networks with non-deterministic event order.

A \emph{run of a network from configuration $c_0$} is a  sequence of configurations $c_0,c_1,c_2,\dots$ such that  $c_{i+1}$ is a successor configuration of $c_{i}$. A \emph{run} is a run from the initial configuration. The set of \emph{reachable configurations from a configuration $c_i$} is the set of all configurations that reside on a run from $c_i$.
The set of \emph{reachable configurations} of a network is the set of reachable configurations from the initial configuration.

\subsection{Safety Verification of Stateful Networks} \label{sec:safety-properties}

In this section we define the \emph{safety} verification problem in stateful networks, as well as the special case of \emph{isolation}.

To describe safety properties, we augment middleboxes with a special \emph{abort state}  that is reached whenever $\delta_m(q,(p,pr)) = \emptyset$, i.e., the forwarding behaviour is undefined (not to be confused with the case where $(\emptyset, q') \in \delta_m(q, (p,\SinglePort))$ for some $q' \in Q_m$).
This lets middleboxes function as ``monitors'' for safety properties.
If $\delta_m(q,(p,pr)) = \emptyset$, and $h \in h(q)$, i.e., $h$ is a history leading to $q$, we say that $m$ \emph{aborts} on $h \cdot (p,\SinglePort)$ (and every extension thereof).
In the symbolic representation, this is captured by the $\babort$ command.

We define \emph{abort configurations} as network configurations where at least one middlebox is in an abort state.

\paragraph{Safety}
The input to the \emph{safety problem} consists of a network $\Net$. 
The output is $\true$ if no abort configuration is reachable in $\Net$, and $\false$ otherwise.

\paragraph{Isolation and Reachability}
An important example of a safety property is isolation.
Informally, isolation is the requirement that certain packets (e.g., packets from a certain host) never reach some host.
In the \emph{isolation problem}, the input is a network $\Net$, a set of hosts $H_i \subseteq H$ and a forbidden set of packets $P_i\subseteq P$. The output is $\true$ if there is no run of $\Net$ in which a host from $H_i$ receives a packet from $P_i$, and $\false$ otherwise.

The isolation problem can be formulated as a safety problem by introducing an \emph{isolation middlebox} $m_{h_i}$ for every host $h_i \in H_i$. The role of $m_{h_i}$ is to monitor all traffic to $h_i$, and abort if a forbidden packet $p \in P_i$ arrives.
All other packets are forwarded to $h_i$. (\figref{Isolation} shows a symbolic representation of such a middlebox connected to $h_i$ on port 0 and to the rest of the network on port 1.)
Clearly, isolation holds if and only if the resulting network is safe.

The \emph{Reachability problem} is the dual of the isolation problem (i.e., the output is flipped).

\begin{figure}
\begin{tabbing}
mm\=mm\=mm\=mm\=mm\=\kill
\> $\brcv(src, dst, tag, prt):$\\
\>\>$prt=0$ $\Rightarrow$ $\bfrw~\{(src,dst, tag,1)\}$\\
\>\>$prt=1 \land (src,dst,tag) ~\bin ~\tforbidden$ $\Rightarrow$ $\babort$\\
\>\>$prt=1 \land \neg((src,dst,tag) ~\bin ~\tforbidden)$ $\Rightarrow$ $\bfrw~\{(src,dst, tag,0)\}$\\
\end{tabbing}
\caption{\label{Fi:Isolation} Isolation checking middlebox.}
\end{figure}

\begin{example} \label{examp:topologies}
\figref{MidTop} shows several examples of interesting middlebox topologies for verification.
In all of the topologies shown we want to verify a variant of the isolation property.
In \figref{benchmark:lbids} we want to verify that $A$, a host, cannot send more than a fixed number of packets to $B$. Here $r_1$ and $r_2$ are rate limiters, \ie they count the number of packets they have seen going from one host to the other, and $lb$ is a load balancer that evenly spreads packets from $A$ along both paths (to minimize the load on any one path). In \figref{benchmark:fwproxy} we want to ensure that host $A$ cannot access data that originates in $S_1$, but should be allowed to access data from $S_2$, where $f$ is a firewall and $c$ is a proxy (cache) server. Finally in \figref{benchmark:mdc} we show a multi-tenant datacenter (e.g., Amazon EC2), where many independent tenants insert rules into firewalls ($f_1$ and $f_2$) and we want to ensure that the overall behaviour of these rules is correct. For example, we would like to ensure that $pri_1$ cannot communicate with $pri_2$, and $pub_2$ communicates with $pri_1$ only if $pri_1$ initiates the connection.
\end{example}

\begin{figure}[t]
    \begin{subfigure}[b]{0.3\textwidth}
        \centering
        \includegraphics[width=0.8\textwidth]{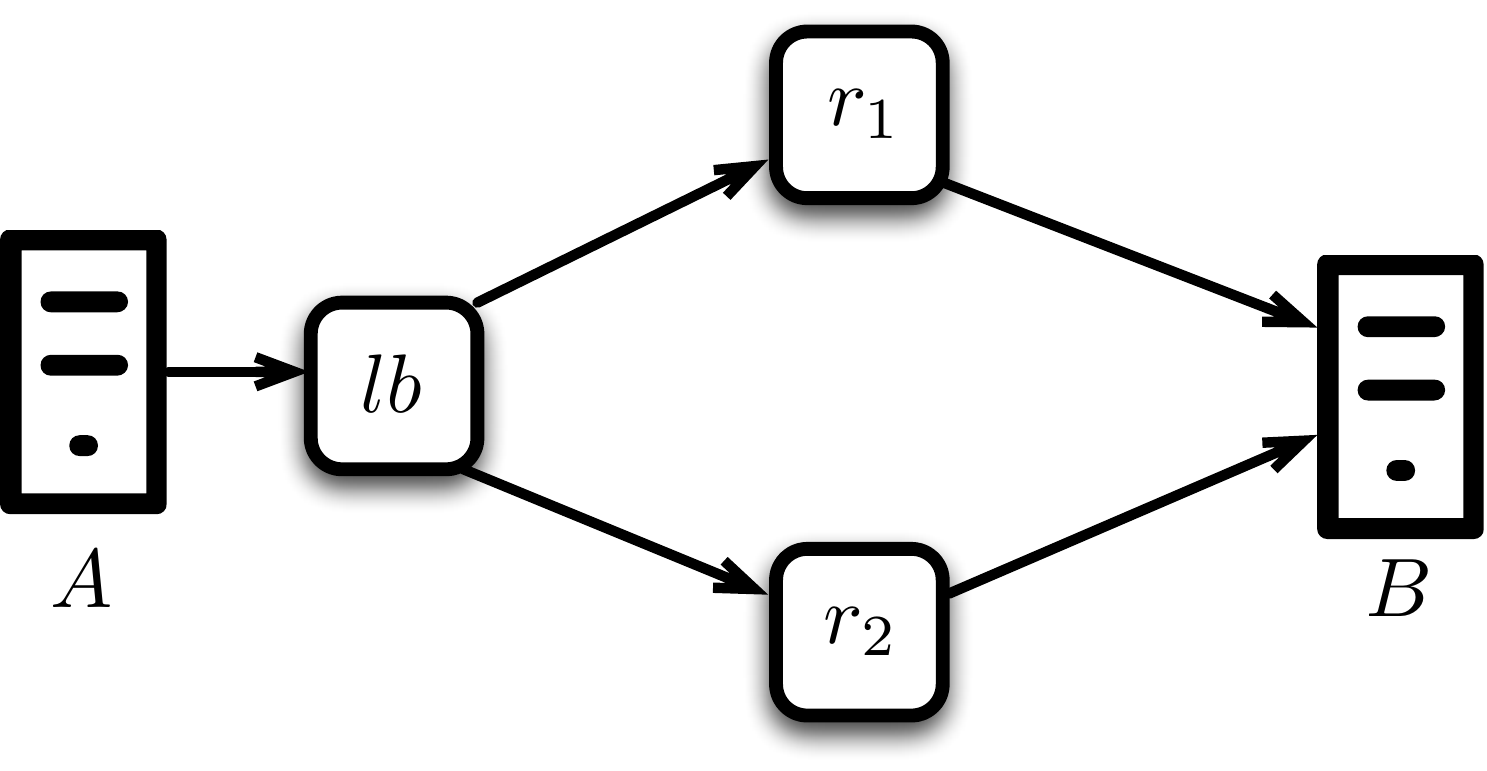}
        \caption{Load Balancer and Rate Limiter}
        \label{Fi:benchmark:lbids}
    \end{subfigure}
    \hfill
    \begin{subfigure}[b]{0.3\textwidth}
        \centering
        \includegraphics[width=0.8\textwidth]{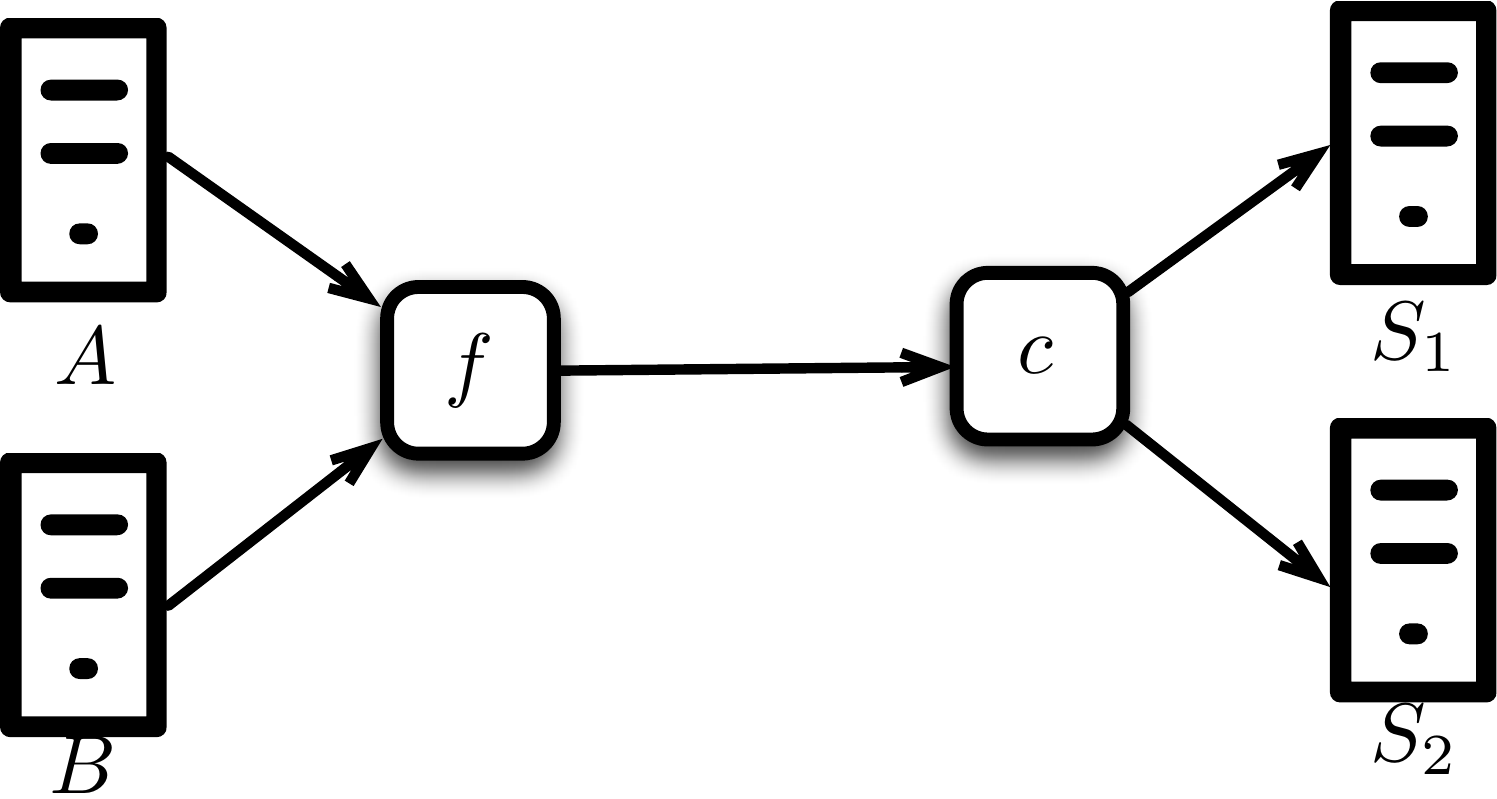}
        \caption{Firewall and Proxy}
        \label{Fi:benchmark:fwproxy}
    \end{subfigure}
    \hfill
    \begin{subfigure}[b]{0.3\textwidth}
        \centering
        \includegraphics[width=0.8\textwidth]{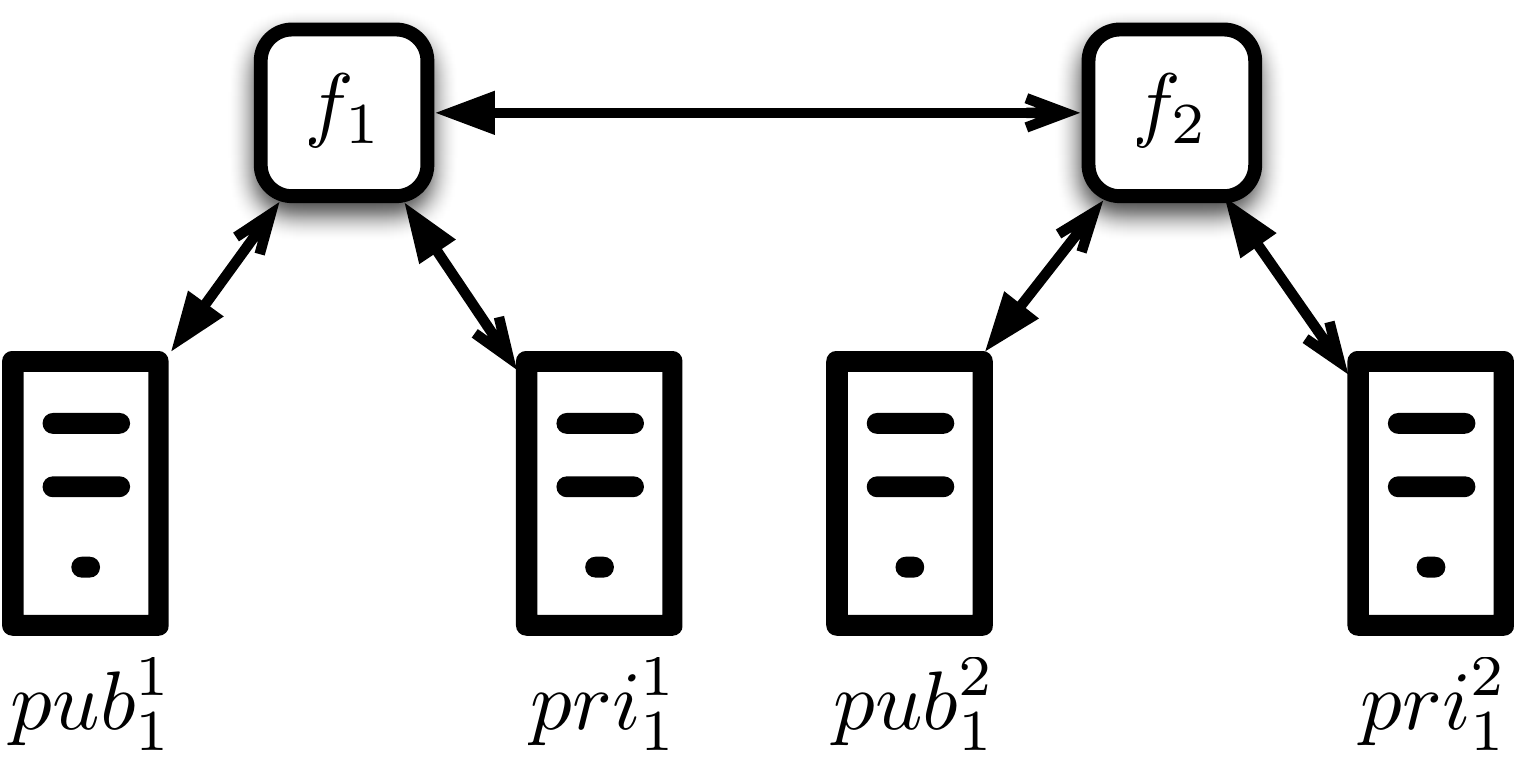}
        \caption{Multi-tenant data center}
        \label{Fi:benchmark:mdc}
    \end{subfigure}
    \caption{Interesting network topologies for verification.}
    \label{Fi:MidTop}
\end{figure}

\subsection{Undecidability of Safety w.r.t. the FIFO Semantics} \label{sec:undecidability}
In this section, we prove undecidability of the safety problem by showing that
(the specific example of) checking isolation w.r.t. the FIFO semantics is
undecidable, even when the network does not have forwarding loops. 
Forwarding loops occur in networks when a packet reaches the same middlebox (or
router) multiple times during the packet's traversal of the network. In the case
of deterministic networks (e.g., stateless networks that consist solely of
routers), forwarding loops result in the packet traversing the network in an
infinite path, never reaching the packet's destination. In general, the
existence of forwarding loops is considered a fault in the network design
\cite{hengartner2002detection}.

It is well known that an automaton with an ordered channel of messages (also known as a \emph{communicating FSM}) can simulate a Turing machine \cite{brand1983communicating}.
This can be used to show that the isolation problem over ordered channels is undecidable in the presence of forwarding loops: a forwarding loop allows a packet to traverse the network and reach the same middlebox any number of times. Therefore, it allows one middlebox in the network to simulate a communicating FSM by having all packets rerouted to it.
However, it turns out that forwarding loops are not the root cause for undecidability.
In this work, we prove that the isolation problem is still undecidable even in the absence of forwarding loops.

To formally define forwarding loops, we augment every packet sent by a host with a unique packet id (e.g., the host id combined with a time stamp). Middlebox forwarding is oblivious to this augmentation: forwarding relations do not depend on the packet id, nor do they change it.
We say that a network has a forwarding loop if there is a run in which a packet with the same packet id is received by the same middlebox twice (i.e.,  a run in which a packet that originates from a middlebox is received by the same middlebox again, possibly after modifications).

We now prove the undecidability result. For completeness of the presentation,
our proof shows a reduction from the halting problem of 2-counter machines
rather than from reachability of communicating FSMs. However, the same idea of
avoiding forwarding loops could be applied to the reduction from commmunicating
FSMs sketched above.

\begin{theorem}\label{thm:OrderedIsUndec}
The isolation problem under the FIFO network semantics is undecidable even for networks with finite-state middleboxes and without forwarding loops.
\end{theorem}

\begin{proof}
We prove undecidability by a reduction from the (undecidable) halting problem of a two-counter machine to the reachability problem, which is the complement of the isolation problem.
A \emph{two-counter machine} $M$ consists of a finite set of control states $Q$, an initial state $q_0\in Q$, a final state $q_f\in Q$, and a set of instructions per state (state transitions).
An instruction determines the next state and manipulates the value of the counters $c_1, c_2$ (initially the value of the two counters is $0$).
An instruction is in one of the two following forms \cite{minsky1961recursive}:
\begin{itemize}
\item $q_1: c_i = c_i + 1$ \lstinline|; GOTO| $q_2$. \\
The instruction increments $c_i$ and changes the state from $q_1$ to $q_2$.
\item $q_1:$ \lstinline|If| $c_i=0$ \lstinline|GOTO| $q_2$ \lstinline|Else| $c_i:= c_i - 1$ \lstinline|; GOTO| $q_3$. \\
The instruction changes the state to $q_2$ if the counter value is zero; otherwise it decrements the counter and goes to state $q_3$.
\end{itemize}

We first describe a reduction that constructs a network with forwarding loops and allows discarding of packets. We then describe how to get rid of the forwarding loops and the discard operation.
Our reduction constructs a network with three middleboxes: a controller middlebox that simulates the state in $Q$, a $c_1$ middlebox that helps simulate the value of the first counter, and a $c_2$ middlebox that helps simulate the value of the second counter, as illustrared in \figref{fifo-undec}.
The network has two hosts: initiator and target.
Intuitively,
the initiator host initiates the simulation of the counter machine,
and the target host receives a packet if and only if the counter machine reaches the final state $q_f$.
Isolation holds if and only if the target host receives no packet.
Both hosts are connected to the controller, which is also connected to $c_1$ and $c_2$.
The set of packet tags is $T = \{\#, 1\}$. Recall that this determines the set of (abstract) packets.
The simulation is done by making sure that the total number of $1$ packets on the ingress and egress channels of each $c_i$ corresponds the value of the simulated counter.

In our construction, the middleboxes decide on forwarding based on the packet tag only.
Middlebox $c_i$ forwards all of its received packets back to the controller host.
We now describe the forwarding behaviour of the controller. Initially, the initiator sends two $\#$ packets to the controller. From that point on, the initiator sends only $1$ packets. 
As our network model does not allow to restrict the order in which hosts send packets, this scheme is enforced by the controller: if any other packet arrives, the controller goes to a sink state in which it discards all received packets. The controller forwards the first $\#$ to $c_1$ and the second $\#$ to $c_2$. When the controller gets a $1$ packet from the initiator it simulates a single step of the counter machine, as follows.
In an increment operation of $c_i$, the controller
sends a $1$ packet to $c_i$.
To simulate a zero test of $c_i$, the controller receives two packets from $c_i$ (if packets from other hosts or middleboxes are received, then the controller goes to a sink state).
If the first received packet is $\#$, then the controller forwards it back to $c_i$. If the second one is also $\#$, then the value of the counter is zero. If it is $1$, then it is discarded (the value of $c_i$ is decremented by $1$). If both packets are $1$, then the first one is discarded and the second is forwarded back to $c_i$.
The simulation of the states of the counter machine is performed by the states of the controller middlebox in a straightforward manner.
Finally, if the controller simulates a transition to $q_f$, then it forwards the packet to the target host.
Hence, the counter machine halts if and only if the target host is not isolated.

\begin{figure}
\begin{center}
\includegraphics[width=0.5\textwidth]{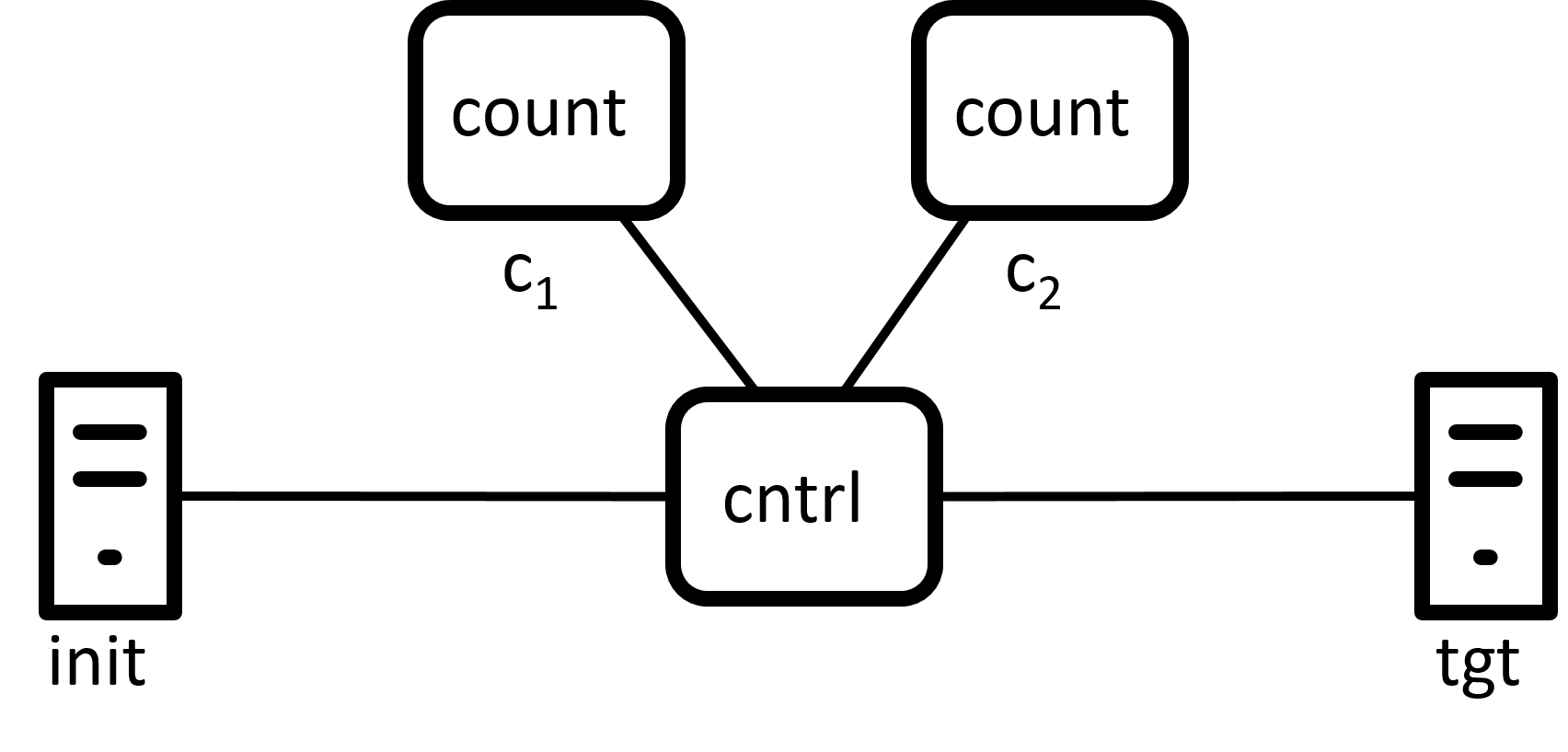}
\end{center}
\caption{The network resulting from the reduction from the halting problem for Two Counter Machines.}\label{Fi:fifo-undec}
\end{figure}

\emph{Construction without discard operation.}
To avoid packet discarding we add a dummy host, and packets that should be discarded are forwarded to the dummy host.

\emph{Construction without forwarding loops.}
To avoid forwarding loops, we add a repeater host to every middlebox.
In the new construction, if a middlebox receives a packet with tag $t$ and needs to forward it to port $p$, then it discards it, and (i)~if the next packet that it receives is not from its repeater with tag $t$, then it goes to a sink state.
(ii)~otherwise, it forwards the packet it got from its repeater to port $p$.
\qed
\end{proof}

\section{Abstract Network Semantics} \label{sec:unordered}

In this section we define an abstract network semantics, called the
\emph{unordered semantics}, which recovers decidability of the safety problem. A
similar setting was explored in \cite{lipton1976reachability,sen2006model} to
recover decidability and obtain the same complexity results we show in
Theorem~\ref{thm:EXPSPACEcomplete}

In the concrete (FIFO) network semantics channels are ordered.
In an ordered channel, if a packet $p_1$ precedes a packet $p_2$ in an ingress channel of some middlebox, then the middlebox will receive packet $p_1$ before it receives packet $p_2$.
We abstract this semantics by an \emph{unordered network semantics}, where the channels are unordered, i.e., there is no restriction on the order in which a middlebox receives packets from an ingress channel.
In this case, the sequence of pending packets in a channel is abstracted by a multiset of packets. Namely, the only relevant information is how many occurrences each packet has in the channel.
The definitions of configurations and runs w.r.t. the unordered semantics are adapted accordingly.
Note that this change does not affect the capacity of the network edges. Consequently the set of network configurations remains infinite.

\begin{remark}
Every run with respect to the FIFO network semantics is also a run with respect to the unordered semantics. Therefore, if safety holds with respect to the unordered semantics, then it also holds for the FIFO semantics, making the unordered semantics a sound abstraction of the FIFO semantics with respect to safety.

The abstraction can introduce false alarms, where a violation exists with respect to the unordered semantics but not with respect to the concrete semantics.
This is demonstrated by \exref{OrderMatters} which presents a network that violates isolation with respect to the unordered semantics, but satisfies isolation with respect to the FIFO semantics.
Still, in many cases, the abstraction is precise enough to enable verification.
In particular, in Lemma~\ref{lem:IncreaseOrderDoesNotMatter} we show that for an important class of networks, the two semantics coincide with respect to safety.

\emph{Lossy channel semantics} is another overapproximation of the FIFO network semantics considered in the literature, where the order on the channels is maintained, but packets can be lost. We note that the unordered semantics also over-approximates the lossy semantics with respect to safety, as any violating run with respect to the lossy semantics can be simulated by a run with respect to the unordered semantics where ``lost'' packets are starved until the violation occurs.
\end{remark}

\begin{example}\label{examp:OrderMatters}
Consider a network with two hosts ($h_1$ and $h_2$), each connected to an authentication middlebox ($m_1$ and $m_2$ respectively), as depicted in \figref{order-matters}. The authentication middleboxes are connected to each other as well.
Each authentication middlebox forwards all packets from a host only if the first packet seen from that host is an authentication key ($k_1$ and $k_2$ for $m_1$ and $m_2$ respectively), otherwise it drops all packets from that host. We would like to verify isolation between $h_1$ and $h_2$. Namely, we would like to verify that no packet with source $h_1$ arrives at $h_2$ and vice versa.

A possible scenario violating isolation w.r.t the unordered semantics is:
(i)~$h_1$ sends $k_1$ and then sends $k_2$;
(ii)~$m_1$ receives $k_1$ and then receives $k_2$ (and forwards both packets in that order).
(iii)~$m_2$ receives $k_2$ before it receives $k_1$ (i.e., the order on the channel between $m_1$ and $m_2$ was not maintained). $m_2$ forwards $k_2$ to $h_2$ and isolation is violated.

On the other hand, if all channels are FIFO, then if $h_1$ first sends $k_2$, it and all subsequent packets from $h_1$ will be dropped by $m_1$. If $h_1$ first sends $k_1$ instead, $m_1$ will forward it to $m_2$, which in turn will drop it and all subsequent packets from $h_1$. Consequently, isolation between $h_1$ and $h_2$ is preserved under the FIFO semantics.
\end{example}

\begin{figure}
\begin{center}
\includegraphics[width=0.5\textwidth]{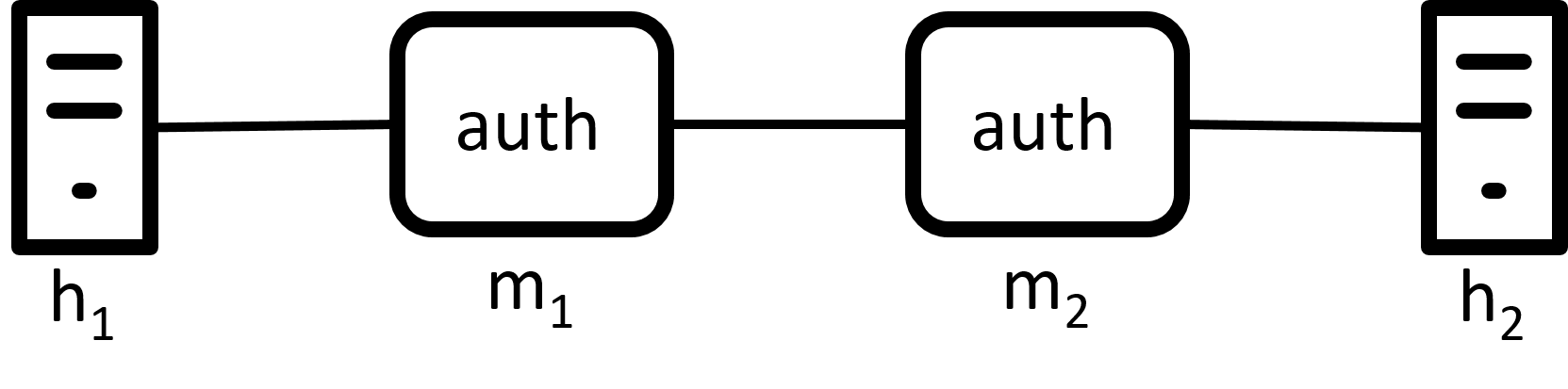}
\end{center}
\caption{A network with two hosts and two authentication middleboxes. Isolation in this network is preserced under the FIFO semantics, but is violated under the unordered semantics.}\label{Fi:order-matters}
\end{figure}

\paragraph{Decidability of Safety w.r.t. the Unordered Semantics}
In the unordered semantics, the network forms a special case of \emph{monotone transition systems}:
We define a partial order $\leq$ between network configurations such that $c_1 \leq c_2$ if the middlebox states in $c_1$ and $c_2$ are the same and $c_2$ has at least the same packets (for every packet type) in every channel.
The network is monotone in the sense that for every run from $c_1$ there is a corresponding run from any bigger $c_2$, since more packets over a channel can only add possible scenarios.
The partial order is trivially a well-quasi-order (as the number of packets cannot be negative), and the predecessor relation is clearly computable.
The classical results of Abdulla et al.~\cite{abdulla1996general} and Finkel et al.~\cite{finkel2001well} prove that in monotone transition systems a backward reachability algorithm always terminates and thus, the safety problem is decidable. Formal arguments and complexity bounds are provided by \thmref{EXPSPACEcomplete}. 
\section{Classification of Stateful Middleboxes}\label{sec:MboxClass}

Encouraged by the decidability of safety w.r.t. the unordered semantics, we are now interested in investigating its complexity.
As a first step, in this section, we identify three special classes of forwarding behaviours of middleboxes within the class of arbitrary middleboxes. Namely, stateless, increasing, and progressing middleboxes. We show that these classes capture the behaviours of real world middleboxes. The classes naturally extend to classes of networks: a network is stateless (respectively, increasing, progressing or arbitrary) if all of its middleboxes are.
As we show in \secref{complexityLowerBounds} and \secref{complexityUpperBounds}, each of these classes results in a different complexity of the safety problem.
Our definitions apply both for finite-state and infinite-state middleboxes.

\paragraph{Stateless Middlebox}
A middlebox $m$ is \emph{stateless} if it can be implemented as a transducer with a single state (in addition to the abort state), i.e., its forwarding behaviour does not depend on its history (with the exception of abort).
Formally, a middlebox $m$ is stateless if for every two histories $h_1,h_2 \in (P\times \PortSet)^*$, packet $p \in P$, port $\SinglePort \in \PortSet$ and output set $o\in 2^{P\times \PortSet}$, $(h_1,(p,\SinglePort),o)\in \ForwardRel$ iff $(h_2,(p,\SinglePort),o)\in \ForwardRel$.

\paragraph{Increasing Middlebox}
A middlebox $m$ is \emph{increasing} if its forwarding relation $\ForwardRel$ is monotonically increasing w.r.t. its history, where histories are ordered by the \emph{subsequence} relation\footnote{A subsequence is a sequence that can be derived from another sequence by deleting some elements without changing the order of the remaining elements.}, denoted by $\sqsubseteq$. Formally, a  middlebox  $m$ is increasing if for every two histories $h_1, h_2 \in (P\times \PortSet)^*$:
if $h_1 \sqsubseteq h_2$, then for every packet $p$, port $\SinglePort$ and output sets $o_1, o_2 \in  2^{P\times \PortSet}$, if
$(h_1, (p,\SinglePort), o_1) \in \ForwardRel$ and $(h_2, (p,\SinglePort), o_2) \in \ForwardRel$ then $o_1 \subseteq o_2$.
Intuitively, this means that new information can only expand the forwarding policy of an increasing middlebox, or lead to an abort.

\begin{remark}
The ``increasing'' property implies that the forwarding relation of an increasing middlebox is in fact a function. Hence, the middlebox can be implemented by a deterministic transducer.
In the following we will refer to the forwarding \emph{function} $\ForwardFunc$ of increasing middleboxes instead of the forwarding relation $\ForwardRel$.
\end{remark}

The following lemma ensures that the behaviour of an increasing middlebox can be precisely captured by a finite-state deterministic transducer.
Its proof uses Higman's lemma~\cite{higman1952ordering} (based on well quasi ordering).

\begin{lemma}\label{lem:IncreasingIsFiniteState}
Any infinite-state increasing middlebox has an implementation as a deterministic finite-state increasing middlebox.
\end{lemma}

\begin{proof}
Consider an infinite-state increasing middlebox $m$, and its forwarding function $\ForwardFunc : (P\times \PortSet)^* \times (P \times \PortSet) \to 2^{P\times \PortSet}$. Recall that $\ForwardFunc$ might be a partial function.

Let $\ForwardFunc(h)$ denote an $\ell\times k$ \emph{output matrix} for the middlebox $m$ and history $h$, where $|\AbsPackets|=\ell$ and $|\PortSet|=k$. We further denote $\AbsPackets=\{p_1,\dots,p_\ell\}$ and $\PortSet=\{\SinglePort_1,\dots,\SinglePort_k\}$. Every entry in the output matrix $\ForwardFunc(h)$ contains the output set for the corresponding pair of packet and port, or $\top$ if it is undefined. Formally $\ForwardFunc(h)_{i,j}=\ForwardFunc(h,(p_i,\SinglePort_j))$ or $\ForwardFunc(h)_{i,j}=\top$ when $\ForwardFunc$ is undefined for the input.

As $\AbsPackets$ and $\PortSet$ are finite, we get that there is a finite number of different output matrices. We denote them by $A_1,\dots,A_n$.
With every output matrix $A_i$ we associate the set of matching histories $h(A_i) = \{h\mid \ForwardFunc(h) = A_i\}$.
Note that $h(A_1) \cup \ldots \cup h(A_n) = (P\times \PortSet)^*$ and that $h(A_i) \cap h(A_j) = \emptyset$ for every $i \neq j$ (since the forwarding function is deterministic). Therefore, for every history $h$ there exists a unique $i$ such that $h \in h(A_i)$.

In the following, we will show that for every $A_i$, the set $h(A_i)$ is regular
and thus we can implement the forwarding function $\ForwardFunc$ of $m$ by using
finite-state automata to recognize the matrix that corresponds to the current
history and then forwarding the current packet accordingly.

We show that for every output matrix $A$, $h(A)$ is regular. 	
We define a partial order $\leq$ over matrices as: $A \leq B$ iff $A_{i,j} \subseteq B_{i,j}$ for every pair of indices $i,j$,
(where $X \subseteq \top$ for every $X \in 2^{P\times \PortSet}$).
We denote by $\mathit{UP}(A)$ the upwards closure of $\{A\}$ with respect to the $\leq$ order on matrices.
We extend the definition of $h(A)$ to sets of matrices: for a (possibly infinite) set of matrices $\mathcal{A}$ we define $h(\mathcal{A}) = \bigcup_{A\in\mathcal{A}}\{h\mid \ForwardFunc(h) = A\}$.
We note that since $m$ is increasing, the set $h(\mathit{UP}(\{A\}))$ is upwards closed with respect to the subsequence relation over histories.
Indeed, if $h_1 \in h(\mathit{UP}(A))$, then $\ForwardFunc(h_1) \geq A$. For every $h_2 \sqsupseteq h_1$, $\ForwardFunc(h_1) \leq \ForwardFunc(h_2)$ (as $m$ is increasing), and thus $\ForwardFunc(h_2) \geq A$, which means that $h_2 \in h(\mathit{UP}(A))$ as well.
Hence, by Higman's lemma and the finite basis property of wqo, we get that $h(\mathit{UP}(A))$ has a finite basis (which consists of histories). We denote the basis $\{h_1,\dots,h_o\}$. Then $h\in h(\mathit{UP}(A))$ if and only if $h\sqsupseteq h_i$ for some $i=1,\dots,o$.

We further observe that for a given history $h_i$, the (infinite) set $\{h\mid h_i\sqsubseteq h\}$ is a regular language, and as regular languages are closed under finite union, we get that the (infinite) set of histories $h(\mathit{UP}(A))$ is regular. Finally, we note that $h(A) = h(\mathit{UP}(A)) \setminus \bigcup \{h(\mathit{UP}(A')) \mid A' \geq A \wedge A' \neq A \}$. Since there are finitely many output matrices, closure properties of regular languages imply that $h(A)$ is regular.

To complete the proof, we describe the transducer contruction. Let $D_i$ be a
finite-state automaton that recognizes $h(A_i)$. We construct a finite-state
transducer $m'$ for $m$, as follows. $m'$ runs $D_1,\ldots,D_n$ in parallel.
They all start from their initial states, and on every new packet $p$ that
arrives from port $\SinglePort$, $m'$ updates the states of $D_1,\ldots,D_n$ in
parallel based on $(p,\SinglePort)$. Exactly one of them, say $D_i$, will reach
an accepting state, in which case $m'$ will process the packet as defined by
$A_i$. Correctness is ensured since for every history $h$, $D_i$ accepts $h$ if
and only if $h \in h(A_i)$, which by definition ensures that $\ForwardFunc(h) =
A_i$. In addition, the construction results in a finite-state transducer since
the number of matrices is finite.

\qed
\end{proof}

\paragraph{Precision of Abstract Semantics in Increasing Networks}
Recall that in general, safety w.r.t. the FIFO semantics and the unordered
semantics do not coincide. However, the following lemmas show that for
increasing networks (with either finite-state or infinite-state middleboxes)
they must coincide, making the abstraction precise for such networks.
Intuitively, this is because in increasing networks if a packet $p$ reaches a
middlebox $m$ once, then unless a middlebox in the network reaches an abort
state, the packet $p$ can reach $m$ again, thus enabling the simulation of
unordered channels with ordered ones. The following lemma formalizes this claim.

\begin{lemma}\label{lem:PacketCanComeAgain}
Let $\Net$ be an increasing network.
For every middlebox $m$, packet $p$ and port $\SinglePort$,
if there exists a run $r$ of $\Net$ from the initial configuration in the FIFO semantics such that in the last step $m$ receives $p$ from $\SinglePort$, then from \emph{any} configuration there exists a run, in the FIFO semantics, that ends in a step in which $m$ receives $p$ from $\SinglePort$ (or in abort).
\end{lemma}

\begin{proof}
We prove the assertion by induction on $|r|$ (the length of the run from the initial configuration).
We fix $m,p,\SinglePort,r$, and an arbitrary configuration $c$ from which we wish to show a run.

If $|r| = 1$, then it must be the case that $m$ received the packet from a neighbor host.
Hence, $c$ has a run in which the same neighbor host sends the same packet to $m$, and after all the previous packets in the ingress channel of $m$ are processed, the packet $p$ arrives from port $\SinglePort$.

If $|r| > 1$, then we consider two distinct cases.
In the first case, the packet was sent to $m$ by a neighbor host, and by the same arguments as before the assertion holds.
In the second case, the packet was sent to $m$ by a neighbor middlebox $m'$.
Let $h'=(p'_1,\SinglePort_1 '),\dots,(p'_n,\SinglePort_n ')$ be the history of packets received by $m'$ before it sent the packet, and let $(p',\SinglePort ')$ be the packet that triggered the forwarding of $p$ from $m'$ to $m$.
Since these packets were received by $m'$ before the last step of $r$ it must be
the case that there exist $n+1$ runs $r_1,\dots,r_{n},r'$ such that run $r_i$
ends when $m'$ receives packet $(p'_i,\SinglePort_i ')$, and run $r'$ ends when
$m'$ receives $(p',\SinglePort ')$. Each of the runs $r_1,\dots,r_{n},r'$ has a
length of at most $|r|-1$, since they are subsequences of the prefix of $r$ that
excludes the packet $(p,\SinglePort)$ sent from $m'$ to $m$.

Hence, by the induction hypothesis there is a run over $\Net$ that begins in $c$ and ends in some configuration $c_1$ after $m'$ received the packet $(p'_1,\SinglePort_1 ')$.
Similarly, for every $i=1,\dots,n$ there is a run that begins in $c_i$ and ends in some configuration $c_{i+1}$ after $m'$ received the packet $(p'_i,\SinglePort_i ')$.
Finally, there is a run from $c_{n+1}$ to a configuration $c'$ that ends after $m'$ received $(p',\SinglePort ')$.
Consider the history $h''$ of $m'$ that is formed in the run $c \leadsto c_1 \leadsto \dots c_{n+1} \leadsto c'$.
Regardless of the history of $m'$ in $c$ (which is the prefix of $h''$), we get that $h'$ is a subsequence of $h''$ (as $(p'_{i+1},\SinglePort_{i+1}')$ is added after $(p'_i,\SinglePort_i ')$).
Hence, after $m'$ receives $(p',\SinglePort ')$, it must forward $p$ to $m$ (due to the fact that $\ForwardFunc_{m'}(h,(p',\SinglePort '))\subseteq \ForwardFunc_{m'}(h'',(p',\SinglePort '))$).
Hence, after $m$ processes all the packets in its ingress channel, it will receive $(p,\SinglePort)$ (or will get to an abort state).
\qed
\end{proof}

In Lemma~\ref{lem:IncreaseOrderDoesNotMatter}, we use the property shown in
Lemma~\ref{lem:PacketCanComeAgain} to prove that any reachable configuration in
an unordered network, is also a reachable configuration of a FIFO network. Given
an unordered violating run, we use the construction described in the proof of
Lemma~\ref{lem:PacketCanComeAgain} to build a FIFO run that ends in the packet
that caused the violation in the unordered run, or in a FIFO violating run in
case the construction from Lemma~\ref{lem:PacketCanComeAgain} resulted in an
abort state.

\begin{lemma}\label{lem:IncreaseOrderDoesNotMatter}
Let $\Net$ be an increasing network.
Then the output of the safety problem in $\Net$ w.r.t. the FIFO semantics and the unordered semantics is identical.
\end{lemma}

\begin{proof}
Recall that any (violating) run w.r.t. the FIFO semantics is also a viable (violating) run w.r.t. the unordered semantics. Therefore, in order to prove the assertion of the lemma, it suffices to prove that for every violating run w.r.t. the unordered semantics there is a violating run w.r.t. the FIFO semantics.

We prove that for every unordered run $r$ and every middlebox $m$ there exists an ordered run $r'$ s.t. $r|_m \sqsubseteq r'|_m$ where $r|_m$ is the history of middlebox $m$ in run $r$.

The proof is by induction on the length of the unordered run $r$.
The base case, where $|r| = 0$, is clear as the history is necessarily empty.

For $|r|>0$, the induction hypothesis guarantees that for the prefix of $r$ of length $|r-1|$, denoted $r_{-1}$, there exists an ordered run $r'_{-1}$ s.t. $r_{-1}|_m \sqsubseteq r'_{-1}|_m$. If $m$ is not the recipient of the last packet, then we consider $r' = r'_{-1}$. The resulting history for middlebox $m$ is $r'|_m=r'_{-1}|_m$, and because $r|_m=r_{-1}|_m$ in this case, we have that $r|_m \sqsubseteq r'|_m$.

If $m$ is the recipient of the last packet, we consider two distinct cases.
In the first case, the final packet $(p,\SinglePort)$ in $r$ was sent by a neighbor host. Since hosts can send packets in any configuration, we append the last event of $r$ to $r'_{-1}$, resulting in the ordered run $r'$. The resulting history for middlebox $m$ is $r'|_m=r'_{-1}|_m\cdot (p,\SinglePort)$, and because $r|_m=r_{-1}|_m\cdot(p,\SinglePort)$, we have that $r|_m \sqsubseteq r'|_m$.

In the second case, the final packet $(p,\SinglePort)$ in $r$ was sent by a neighbor middlebox $m'$. We consider the history of middlebox $m'$ for $r_{-1}$ --- the prefix of $r$ of length $|r-1|$, denoted $h=r_{-1}|_{m'}=\langle(p_0,\SinglePort_0),\cdots,(p_l,\SinglePort_l)\rangle$. By the induction hypothesis, there exists an ordered run $r''_{-1}$ s.t. $r_{-1}|_{m'} \sqsubseteq r''_{-1}|_{m'}$, and by \lemref{PacketCanComeAgain} we get that for every packet $(p_i,\SinglePort_i)$ in $h$ from any configuration there exists an ordered run that ends in middlebox $m'$ receiving $(p_i,\SinglePort_i)$, or there exists a run that leads to a safety violation (in which case we have reached the goal of this construction and are done).

We proceed by constructing the run $r'$. We first construct the run $\tilde{r} = r'_{-1}\cdot r'_{1}\cdots r'_{l}$ where $r'_{-1}$ is the ordered run guaranteed by the induction hypothesis s.t. $r_{-1}|_m \sqsubseteq r'_{-1}|_m$, and
$r_i$ is the ordered run ending in the middlebox $m'$ receiving the packet $(p_i,\SinglePort_i)$, starting from the configuration at the end of the previous run.
The construction ensures that $r_{-1}|_m \sqsubseteq \tilde{r}|_m$ (since $r_{-1}|_m \sqsubseteq r'_{-1}|_m$). In addition, because $r_{-1}|_{m'} =\langle(p_0,\SinglePort_0),\cdots,(p_l,\SinglePort_l)\rangle \sqsubseteq \tilde{r}|_{m'}$ 
and $m'$ is increasing, $m'$ can send the packet $(p,\SinglePort)$ to $m$ after $\tilde{r}$. We obtain $r'$ by appending to $\tilde{r}$ the final event of $r$, where $m'$ sends the packet $(p,\SinglePort)$ to $m$.
Since $r'|_m =  \tilde{r}|_m \cdot (p,\SinglePort)$, $r|_m =  r_{-1}|_m \cdot (p,\SinglePort)$ and $r_{-1}|_m \sqsubseteq \tilde{r}|_m$, we get that $r|_m \sqsubseteq r'|_m$.

In particular, we can construct an ordered run in which $m$ has an aborting history.
\qed
\end{proof}

\paragraph{Progressing Middlebox}
In order to define progressing middleboxes, we define an equivalence relation between middlebox states based on their forwarding behaviour. States $q_1,q_2$ are equivalent, denoted $q_1 \approx q_2$, if $L(q_1) = L(q_2)$.
A middlebox $m$ is \emph{progressing} if it can be implemented by a transducer in which whenever the state is changed into a non-equivalent state, it will never return to an equivalent state. Formally, if $(o',q') \in \delta_m(q,(p,\SinglePort))$ and $q' \not\approx q$ (where $q,q'$ are reachable states of $m$) then for any future sequence of packets $h\in (P\times \PortSet)^*$, if $(\gamma'', q'') \in \delta_m(q',h)$ for some $\gamma''$ and $q''$, then $q'' \not\approx q$.

As opposed to increasing middleboxes, progressing middleboxes might require infinitely many states. In this case nondeterminism is essential as it allows  to support the abstraction of infinite-state middleboxes via finite-state transducers.
\begin{example}[Infinite-state progressing middlebox]\label{examp:InfMemoryEnentuallyStateless}
Consider the packet space $H\times H \times \{0,1\}$, and a deterministic middlebox $m$ with a single port whose forwarding function is defined as follows.
As long as all received packets have tag $0$, then each packet is forwarded (as is) back to the single port. When a packet with tag $1$ arrives for the first time, if the number of previous packets is prime, then all future packets are discarded.
Otherwise, all future packets are forwarded back to the single port.
Prime numbers are not recognizable by finite-state machines. Hence, there is no finite-state implementation of $m$.
On the other hand, $m$ is progressing since its state always progresses (from counting to always discarding or always forwarding).
\end{example}

Finite-state progressing middleboxes have the following useful property:
\begin{lemma}\label{lem:EventuallyStatelessDAG}
Every finite-state progressing middlebox has an implementation as a finite-state transducer whose underlying state graph has a tree structure, except for, possibly, self-loops.
\end{lemma}

\begin{proof}

We show an implementation as a directed acyclic graph (DAG), possibly with self loops. The transformation to a tree is then straightforward.
Let $m$ be a minimal transducer that implements the progressing middlebox.
We consider the language $L(q)$ of each state $q$ in $m$.
Minimality ensures that no two states in $m$ have the same language (otherwise they are equivalent and can be merged). Therefore, each state $q$ represents a \emph{unique} language $L(q)$.

Towards a contradiction we assume that there is a directed loop that is not a self-loop in $m$.
A loop implies that there are two states $q_1 \not \approx q_2$ in $m$ such that $q_1$ transitions to $q_2$ by some sequence $h_2$ and $q_2$ transitions back to $q_1$ by some sequence $h_3$. Further, by minimality of $m$, $q_1$ is reachable by some sequence $h_1$.

Since $m$ is progressing, contradiction is obtained.
\qed
\end{proof}

The next lemma summarizes the hierarchy of the classes (as illustrated by Figure~\ref{Fi:Complexity}).
\begin{lemma}\label{lem:Hierarchy}
\begin{itemize}
\item Any stateless middlebox is also increasing.
\item Any increasing middlebox is also progressing.
\end{itemize}
\end{lemma}

\begin{proof}
The first part of the lemma is straightforward.

Consider the second part of the lemma.
Let $m$ be a deterministic transducer of an increasing middlebox and $\ForwardFunc$ is its forwarding function.
Towards a contradiction assume that $m$ is not progressing, i.e. there exist two states $q_1 \not \approx q_2$ and three histories $h_0,h_1,h_2$ s.t. $(\gamma_0,q_1) \in \delta_m(q^0, h_0)$, $(\gamma_1,q_2) \in \delta_m(q^0, h_0\cdot h_1)$ and $(\gamma_2,q_1) \in \delta_m(q^0, h_0\cdot h_1\cdot h_2)$.

Because $q_1 \not \approx q_2$, there exists a history $h$  s.t. $\ForwardFunc(h_0 \cdot h) \neq \ForwardFunc(h_0\cdot h_1 \cdot h)$, and since $m$ is increasing it must be the case that
$\ForwardFunc(h_0 \cdot h) \subset \ForwardFunc(h_0\cdot h_1 \cdot h)$.

However, since $m$ is deterministic and $h_0$ and $h_0\cdot h_1\cdot h_2$ lead to the same state, namely $q_1$, it must be that $\ForwardFunc(h_0 \cdot h)=\ForwardFunc(h_0\cdot h_1\cdot h_2 \cdot h)$ and we get that
$\ForwardFunc(h_0\cdot h_1 \cdot h) \supset \ForwardFunc(h_0\cdot h_1\cdot h_2 \cdot h)$, in contradiction to the fact that $m$ is increasing.
\qed
\end{proof}

\subsection{Syntactic Characterization of Middlebox Classes} \label{sec:classes-symbolic}
The classes of middleboxes defined above can be characterized  via syntactic
restrictions on their symbolic representation. In \secref{complexityUpperBounds}
we will use the syntactic characterization to obtain more realistic complexity
upper bounds, stated in terms of the symbolic representation rather than the
explicit state-space of middleboxes.

A middlebox representation is \emph{syntactically stateless} if it does not use any insert or remove command on any relation.
A middlebox representation is \emph{syntactically increasing} if it does not use the remove command on any relation,
does not use negated membership predicates in the guards
and all guards are mutually exclusive (i.e. no two guards can be \emph{true} at the same time).
A middlebox representation is \emph{syntactically progressing} if it does not use the remove command on any relation.

\begin{lemma}\label{lem:StatelessIffAMDLStateless}
Every stateless finite-state middlebox has an equivalent syntactically stateless symbolic representation and vice versa.
\end{lemma}

\begin{proof}
The lemma is trivial for stateless middleboxes, as both the transducer and the symbolic representation simply describe a fixed forwarding table.
\qed
\end{proof}

\begin{lemma}\label{lem:IncreasingIsAMDLIncreasing}
Every increasing finite-state middlebox has an equivalent syntactically increasing symbolic representation and vice versa.
\end{lemma}
\begin{proof}
We first show that every increasing finite-state middlebox has an equivalent syntactically increasing symbolic representation.
Let $m$ be an increasing finite-state middlebox implemented by a deterministic transducer with state set $Q=\{q_0,\dots,q_n\}$, where $q_0$ is the initial state. By Lemma~\ref{lem:Hierarchy} and Lemma~\ref{lem:EventuallyStatelessDAG} we may assume w.l.o.g that the underlying graph of $m$ is a tree.
We construct a symbolic program $A$ with one unary relation $R$ over the constants $q_{-1},q_0,\dots,q_n$. Initially $R = \{q_0\}$.
To describe $A$ we need the next three notations. To reduce the notational burden, we use packets $p$ instead of pairs $(p,pr)$ of a packet and an input port.
For a state $q_i$ and a packet $p$ we denote the successor state of $q_i$ according to packet $p$ by $q_i\to_p$ (we note that possibly $q_i\to_p = q_i$).
The successor state is unique since the transducer is deterministic.
We denote by $q_i(p)$ the output of $m$ when $m$ is in state $q_i$ and packet $p$ is received.
We denote the (single) predecessor of $q_i$ in the tree by $\mathit{pre}(q_i)$
(we note that in case the state $q_i$ has a self loop, the predecessor function
returns the unique predecessor of $q_i$ that is not $q_i$. i.e.,
$\mathit{pre}(q_i)=q_j$ s.t. $q_i\neq q_j$).
For uniformity, we assume that the root $q_0$ also has a predecessor, namely, $q_{-1}$ with $q_{-1}(p) = \emptyset$ for every packet $p$.

We now describe how $A$ processes a packet $p$:
\begin{itemize}
\item \emph{Relation update.} For every $q_i \in R$: insert $q_i\to_p$ to $R$.
\item \emph{Output.} For every $q_i \in R$: output $q_i(p) \setminus \mathit{pre}(q_i)(p)$.
\end{itemize}
We first observe that $A$ can be implemented as a syntactically increasing program.
Indeed, the ``for every'' loops can be replaced by a sequential composition of finitely many guarded commands
consisting of positive relation membership queries, and only $\binsert$ update operations.

We now show that the forwarding behaviours of $A$ and $m$ are identical and hence $A$ is indeed a correct symbolic representation of $m$.
Let $h$ be an arbitrary history and let $p$ be an arbitrary packet.
By a simple induction we get that the states in the relation $R$ are exactly the states that $m$ visited while processing the history $h$.
We assume w.l.o.g that the set of visited states (after history $h$) is $\{q_0,\dots,q_k\}$ and that $q_i = \mathit{pre}(q_{i+1})$.
We prove, by induction on $k$, that the outputs of $m$ and $A$ are identical.
In the base case $k=0$, and the proof follows as we defined $\mathit{pre}(q_{0})(p)=q_{-1}$ and $q_{-1}(p)=\emptyset$.
For $k>1$, we observe that since $m$ is increasing and a prefix is also a subsequence, then $q_{k-1}(p)\subseteq q_{k}(p)$.
Hence, $q_{k}(p) = (q_{k}(p) \setminus q_{k-1}(p)) \cup q_{k-1}(p)$.
By the induction hypothesis, we get that $A$ first outputs $q_{k-1}(p)$, and by the implementation of $A$, we get that it then outputs $q_{k}(p) \setminus q_{k-1}(p)$.
Hence, overall $A$ outputs $q_k(p)$, and the proof of the claim is complete.

To conclude, we proved that $A$ is a syntactically increasing symbolic representation of $m$.

For the converse direction, we show that the forwarding behaviour of a middlebox given via a syntactically increasing symbolic representation is increasing.
Let $A$ be a syntactically increasing symbolic program. For simplicity we assume that $A$ has only one relation $R$.
The mutually exlusive guard requirement implies determinstic execution. Consequently, for a history $h$ we can denote by $R^{h}$ the unique content of relation $R$ after $h$.
We claim that if $h_1 \sqsubseteq h_2$, then $R^{h_1}\subseteq R^{h_2}$.
The proof follows from the fact that all the guards in $A$ have positive conditions and from the fact that elements are only added to the relation.
As the forwarding behaviour depends only on the state of the relation, and since all conditions are positive, we get that the forwarding behaviour is increasing.
\qed
\end{proof}

\begin{lemma}\label{lem:AMDLCyclicIffAcyclic}
Every progressing finite-state middlebox has an equivalent syntactically progressing symbolic representation and vice versa.
\end{lemma}
\begin{proof}
We first show that every progressing finite-state middlebox has an equivalent syntactically progressing symbolic representation.
Let $m$ be a progressing finite-state middlebox, and by Lemma~\ref{lem:EventuallyStatelessDAG} we may assume w.l.o.g that the underlying state graph of $m$ is a tree.
Let $Q = \{q_0,\dots,q_n\}$ be the states of $m$.
We construct a symbolic program $A$ similarly to the construction in the proof of \lemref{IncreasingIsAMDLIncreasing}
(with one unary relation $R$ over the constants $q_0,\dots,q_n$, where initially $R = \{q_0\}$, and where $R$ accumulates the traversed states).
When a packet $p$ is processed, the program identifies the current state by computing a maximal (according to topological order) state $q_i$ in $R$ (this is implemented using a guard for every path from the tree root to each state in the state tree).
It then adds $q_i\to_p$ to $R$ and outputs $q_i(p)$.
Since $m$ is a tree,  there always exists exactly one maximal state in $R$, and we get that $A$ always simulates $m$ correctly.

For the converse direction, we show that the forwarding behaviour of a middlebox given via a syntactically progressing symbolic representation is progressing.
Let $A$ be a syntactically progressing symbolic program. For simplicity we assume that $A$ has only one relation $R$.
We recall that the domain of $R$ is always finite, and thus it has only a finite number of different states (interpretations).
We construct a middlebox $m$ whose states are exactly the states of $R$, and the forwarding function is exactly according to those states.
As $A$ is progressing, we get that elements are only added to $R$, and thus the underlying graph of $m$ is progressing.
\qed
\end{proof}

\subsection{Examples}\label{sec:Examples}
In this section, we introduce several middleboxes, each of which resides in one of  the classes of the hierarchy presented above.

\paragraph{ACL Switch}
An \emph{ACL switch} has a fixed access control list (ACL) that indicates which packets it should forward and which packets it should discard.
Typically the rules in the list refer to the port number or to hosts that are allowed to use a certain service.
As such, the forwarding policy of an ACL switch is based only on the source host and/or ingress port of the current packet, and does not depend on previous packets.
Hence, an ACL switch can be implemented by a stateless middlebox.

\paragraph{Hole-Punching Firewall}
A \emph{hole-punching firewall} is described in Example~\ref{Ex:firewall}.
As the set of trusted hosts depends on the history of the middlebox, a hole punching firewall cannot be captured by a stateless middlebox. (Formally, given two different histories, the forwarding function might produce a different output for the same packet and port.)

However, a hole punching firewall is an increasing middlebox. This follows since for every source host $s$ and two histories $h_1\sqsubseteq h_2$, if $s$ is trusted according to $h_1$, then it is also trusted according to $h_2$.
The proof of the latter is by induction on $|h_1|$.
In the base case $|h_1|=0$, and therefore $s$ is in the initial list of trusted hosts (and therefore, it is trusted also in $h_2$).
If $|h_1|>0$, then $h_1 = h_1' \cdot (p,pr)$. We consider two distinct cases:
In the first case $s$ was trusted before the last packet $p$ in $h_1$ was received.
Hence, by the induction hypothesis we get that $s$ is trusted also in $h_2$.
In the second case $s$ became trusted only after the last packet $p$ was processed.
In this case, $p$ had a trusted source host $s_1$ (according to $h_1'$) with destination $s$. Since $h_1\sqsubseteq h_2$, there exist $h_2', h_2''$ such that $h_2 = h_2' \cdot (p,pr) \cdot h_2''$ and $h_1' \sqsubseteq h_2'$.
By the induction hypothesis, the source host $s_1$ of the last packet $p$ is also trusted according to $h_2'$, and therefore $s$ is trusted also in $h_2' \cdot (p,pr)$. As the set of trusted hosts never decreases, $s$ remains trusted in $h_2$.

\paragraph{Learning Switch}
A \emph{learning switch} dynamically learns the topology of the network and constructs a routing table accordingly.
Initially, the routing table of the switch is empty.
For every host $h$ the switch remembers the first port from which a packet with source $h$ has arrived.
When a packet arrives, if the port of the destination host is known, then the packet is forwarded to that port;
otherwise, the packet is forwarded to all connected ports excluding the input-port.

A learning switch is a progressing middlebox. Intuitively, after the middlebox's forwarding function has changed to incorporate the destination port for a certain host $h$, it will never revert to a state in which it has to flood a packet destined to $h$. A learning switch is however, not an increasing middlebox, as packets destined to a host whose location is not known are initially flooded, but after the location of the host is learned, a single copy of all subsequent packets is sent.

\figref{Learning} depicts a symbolic representation of a learning switch that uses a binary relation \tconnected\ storing connections between hosts and ports.
If the port of the destination host is known, then the packet is forwarded to that port;
otherwise, the packet is forwarded to all connected ports excluding the input-port.
The last command in the program is a syntactic shorthand used to avoid the explicit enumeration of incoming ports required to correctly perform the flood operation.

\begin{figure}
\begin{tabbing}
mm\=mm\=mm\=mm\=mm\=\kill
\> $\brcv(src, dst, tag, prt):$\\
\>\>$\neg~((dst, prt)~\bin~\tconnected)~\Rightarrow$\\
\>\>\>$\tconnected.\binsert (src, prt);$ // remember src's port\\
\>\>$(dst, 1)~\bin~\tconnected~\Rightarrow~\bfrw~ \{(src,dst,tag,1)\}$\\
\>\>$(dst, 2)~\bin~\tconnected~\Rightarrow~\bfrw~ \{(src,dst,tag,2)\}$\\
\>\>$(dst, 3)~\bin~\tconnected~\Rightarrow~\bfrw~ \{(src,dst,tag,3)\}$\\
\>\>$\neg~((dst, 1)~\bin~\tconnected)~\land~$\\
\>\>\>$\neg~((dst, 2)~\bin~\tconnected)~\land~$\\
\>\>\>$\neg~((dst, 3)~\bin~\tconnected)~\Rightarrow$\\
\>\>\>\>$\bfrw~\{(src,dst,tag,oprt) \mid oprt~\bin~\tallPorts~\band~oprt \neq prt\}$ 	// flood\\
\end{tabbing}
\caption{\label{Fi:Learning}%
A learning switch with three ports.}
\end{figure}

\paragraph{Proxy Server}
The \emph{Proxy server} as described in Example~\ref{Ex:proxy} is a progressing middlebox. After it has stored a response, it nondeterministically replies with the stored response, or sends the request to the server again.
Once a new request is responded by a proxy the forwarding behaviour changes as it takes into account the new response, and it never returns to the previous forwarding behaviour (as it does not ``forget'' the response).
This example demonstrates how nondeterminism is used to model middleboxes whose concrete behaviour depends on packet payloads. In a concrete network model that does not abstract away the packet payload, the proxy middlebox would always reply to a request with a stored response and never forward it to the server.

\paragraph{Round-Robin Load Balancer}
A \emph{load balancer} is a device that distributes network traffic across a number of servers.
In its simplest implementation, a round-robin load balancer with $n$ out-ports (each connected to a server) forwards the $i$-th packet it receives to out-port $i\pmod n$.
Round-robin load balancers are not progressing middleboxes, as the same forwarding behaviour repeats after every cycle of $n$ packets.

\figref{Load} depicts a symbolic representation of a round-robin load balancer with 3 ports: port 0 is an `input' port, and ports 1 and 2 are `output' ports on which the load balancer splits the incoming traffic. It uses a unary relation {\tnextport} to hold the port to which the next packet is to be sent.

\begin{figure}
\begin{tabbing}
mm\=mm\=mm\=\kill
\> $\brcv(src, dst, tag, prt):$\\
\>\>$(prt=0\land(1)~\bin~\tnextport)~\Rightarrow~\bfrw~ \{(src,dst,tag,1)\}$\\
\>\>$(prt=0\land(1)~\bin~\tnextport)~\Rightarrow~\tnextport.\bremove~ 1$\\
\>\>$(prt=0\land(1)~\bin~\tnextport)~\Rightarrow~\tnextport.\binsert~ 2$\\
\>\>$(prt=0\land(2)~\bin~\tnextport)~\Rightarrow~\bfrw~ \{(src,dst,tag,2)\}$\\
\>\>$(prt=0\land(2)~\bin~\tnextport)~\Rightarrow~\tnextport.\bremove~ 2$\\
\>\>$(prt=0\land(2)~\bin~\tnextport)~\Rightarrow~\tnextport.\binsert~ 1$\\
\>\>$(prt=1\lor~prt=2)~\Rightarrow~\bfrw~ \{(src,dst,tag,0)\}$\\
\end{tabbing}
\caption{\label{Fi:Load}%
A 3-port round-robin load-balancer.}
\end{figure}

\begin{remark}
In practice, middlebox behaviour can also be affected by timeouts and session termination. For example, in a firewall, a trusted host may become untrusted when a session terminates (which makes the firewall behaviour no longer increasing).
Similarly, cached content of a cache server expires after a certain period of time (which violates progress).
In this work, we do not model timeouts and session termination.
\end{remark} 
\section{Lower Bounds on Complexity of Safety w.r.t. the Unordered Semantics} \label{sec:complexityLowerBounds}

When considering the unordered network semantics, the safety problem becomes decidable for networks with finite-state middleboxes.
In this section, we analyze its complexity lower bounds.
The complexity bounds are w.r.t the input size, namely,
(i)~the number of hosts; (ii)~number of middleboxes; and (iii)~the encoding size of the middleboxes functionality, i.e., the size of the explicit state machine (if the encoding is explicit) or the number of characters in the symbolic representation (if the encoding is symbolic).

In \secref{complexityUpperBounds} we present matching upper bounds for networks represented symbolically. Since symbolic representations are at least as succinct as  explicit-state descriptions of finite-state middleboxes, all the lower bounds obtained for the explicit finite-state model apply for the symbolic one as well, and all the upper bounds obtained for the symbolic model are applicable to the explicit finite-state model, resulting in tight complexity bounds, both for explicit finite-state middleboxes and for symbolic ones.

We obtain lower bounds for the safety verification problem by considering the isolation problem. Recall that the isolation problem reduces to a safety problem by the introduction of isolation middlebox. Since isolation middleboxes are stateless, they do not change the class of the input network. We can therefore  deduce that the same lower bounds also hold for the more general safety problem.

\subsection{Unordered Safety in Progressing Networks is coNP-hard.}
\begin{lemma}\label{lem:progressingIsNPHard}
The isolation problem w.r.t. the unordered network semantics for a progressing network is coNP-hard.
\end{lemma}
\begin{proof}
We show a reduction  from the (NP-hard) Hamiltonian Path problem to the reachability problem, which is the complement of the isolation problem.
Recall that the Hamiltonian Path problem is given a directed graph $G(V,E)$, a source vertex $s \in V$ and a target vertex $t \in V$, and it determines whether there is a simple path from $s$ to $t$ in $G$ with length $|V|$.

In the reduction, we use
\emph{flood-once} middleboxes that upon receiving a packet with a numeric tag (from a finite domain) increment the packet tag and flood the new packet. All following packets that arrive at the middlebox are discarded. These flood-once middleboxes are finite-state progressing middleboxes.

We construct a network with a single flood-once middlebox for every vertex in the graph and connect them in accordance with the edges in the graph.
In addition, we create two hosts $h_{s}$ and $h_{t}$ and connect them to the middleboxes representing the source and target in the graph.
We use packet tags $\{0,\ldots,|V|\}$.
Host $h_{s}$ sends packets with tag $0$. The reachability problem is to determine whether $h_{t}$ can receive a packet with tag $|V|$.

The flood once middleboxes ensure that the packet tags `count' the length of the path. Thus, a Hamiltonian Path corresponds to a packet with the tag $|V|$ arriving at the destination host $h_{t}$, and the correctness of the reduction follows.
\qed
\end{proof}

The following lemma shows that a similar result can be obtained using more ``standard'' middleboxes, namely, stateless middleboxes and learning switches.
\begin{lemma}\label{lem:LearningSwitchIsNPHard}
The isolation problem w.r.t. the unordered network semantics for a network where each middlebox is either stateless or a learning switch is coNP-hard.
\end{lemma}

\begin{proof}
The proof is by reduction from the (NP-hard) Hamiltonian Path problem to the reachability problem.
We use the same notation as used in the proof of \lemref{progressingIsNPHard}.
W.l.o.g we assume that the out-degree of all vertices of the directed graph $G$ is two.
For the reduction, we construct a network with three hosts, namely, $h_s,h_t$ and $h_d$, and $4|V|$ middleboxes, 
as described below. The topology of the resulting network is illustrated in \figref{hamiltonian-path-red}.
The set of packet tags is $\{0,\dots,|V|\}$.
As before, the reachability problem is to determine whether host $h_{t}$ can receive a packet with tag $|V|$.

We now describe the network in more detail. With every vertex $v$ we associate three stateless middleboxes, namely, $v_A,v_B$ and $v_C$, and a learning switch $v_{\mathit{LS}}$, illustrated in \figref{widget}. Intuitively, these middleboxes will simulate a ``flood once'' middlebox.
The middlebox $v_A$ is connected to $v_B$, $v_C$ and $v_{\mathit{LS}}$.
The middlebox $v_{\mathit{LS}}$ is connected to $v_B$ and $v_C$ as well as to $v_A$, and if $(v,u_1)\in E$ and $(v,u_2)\in E$, then $v_B$ has a link to $(u_1)_A$ and $v_C$ is connected to $(u_2)_A$.
Host $h_s$ is connected to $(v_s)_A$ and is allowed to send only the packet $(h_s,h_t,0)$ (source $h_s$, destination $h_t$, and tag $0$).
Host $h_t$ is connected to 
$(v_t)_B$ and $(v_t)_C$.
Host $h_d$ is a dummy host, disconnected from any middlebox. Its purpose is merely to allow three distinct host ids.
The forwarding function of the learning switch is as described in \secref{Examples}.
The forwarding function of the stateless middleboxes is defined as follows:
\begin{itemize}
\item packets received by $v_A$ from some $u_B$ or $u_C$: if the packet is $(h_s,h_t,n)$, namely, source is $h_s$, destination is $h_t$ and packet tag is $n$, then forward it to $v_{\mathit{LS}}$.
\item packets received by $v_A$ from $v_{\mathit{LS}}$: if the packet is $(h_d,h_s,n)$, then forward packet $(h_t,h_d,n)$ to $v_{\mathit{LS}}$.
\item packets received by $v_B, v_C$ from $v_{\mathit{LS}}$: if the packet is $(h_s,h_t,n)$ forward packet $(h_d,h_s,n)$ to $v_{\mathit{LS}}$.
If the packet is $(h_t,h_d,n)$ forward packet $(h_s,h_t,n+1)$ to the appropriate $u_A$.
Otherwise, discard.
\end{itemize}
All other packets are discarded.

We first give an informal description of how a packet is processed and then turn to formally prove the correctness of the reduction.
When $v_A$ receives a $(h_s,h_t,n)$ packet from some $u_B$ or $u_C$ it sends it to the learning switch.
When $v_{\mathit{LS}}$ first receives the packet it forwards it to all of its neighbors except for $v_A$ (from which it was received) and marks the port connected to $v_A$ as the destination port to $h_s$.
$v_B$ and $v_C$ reply with $(h_d,h_s,n)$, and when the first of these packets arrives to $v_{\mathit{LS}}$, then it marks either $v_B$ or $v_C$ as the destination of $h_d$.
In addition, as the port connected to $v_A$ is marked as the destination to $h_s$, the learning switch sends the packets $(h_d,h_s,n)$ to $v_A$.
$v_A$ responds with $(h_t,h_d,n)$.
When $v_{\mathit{LS}}$ receives the packet it marks the port connected to $v_A$ as the destination for $h_t$ and forwards the packet to $v_B$ or $v_C$ (depending on which was marked as the destination for $h_d$).
$v_B$ or $v_C$ increments the tag and forwards the packet to a neighbor $u_A$.
All additional packets of the form $(h_s,h_t,n')$ that will arrive to $v_A$ after $v_B$ or $v_C$ has already incremented the tag will be forwarded by $v_{\mathit{LS}}$ back to $v_A$ (as it was marked as the destination port to $h_t$), and in $v_A$ they will be discarded.

We now give a formal proof.
We claim two assertions:
(i)~For every $v\in V$, at most one of the middleboxes $v_B$ and $v_C$ forwards a packet to an adjacent node (other than $v_{\mathit{LS}}$).
(ii)~Both $v_B$ and $v_C$ will never forward the same packet twice.
The proof of item~(i) is due to the fact that every packet passes through the learning switch and the learning switch will mark only one of $v_B$ or $v_C$ as the destination of $h_d$.
The proof of item~(ii) is due to the fact that if a packet $p$ is generated as a result of $v_B$ ($v_C$) sending a packet to an adjacent middlebox, then at this stage $v_A$ is already marked by the learning switch as the destination of $h_t$.
Therefore, when the packet $p$ reaches $v_A$, it will be forwarded from the learning switch back to $v_A$ and will be discarded.
Hence, it can never reach $v_B$ ($v_C$) again.
By the two assertions we get that reachability holds if and only if a packet visited $|V|$ different middleboxes $(v_1)_{X_1},\dots,(v_{|V|})_{X_{|V|}}$ for $X_i \in \{B,C\}$, and each such middlebox was visited exactly once.
Hence, reachability holds iff a Hamiltonian path exists.
\qed
\end{proof}

\begin{figure}
\begin{center}
\includegraphics[width=0.5\textwidth]{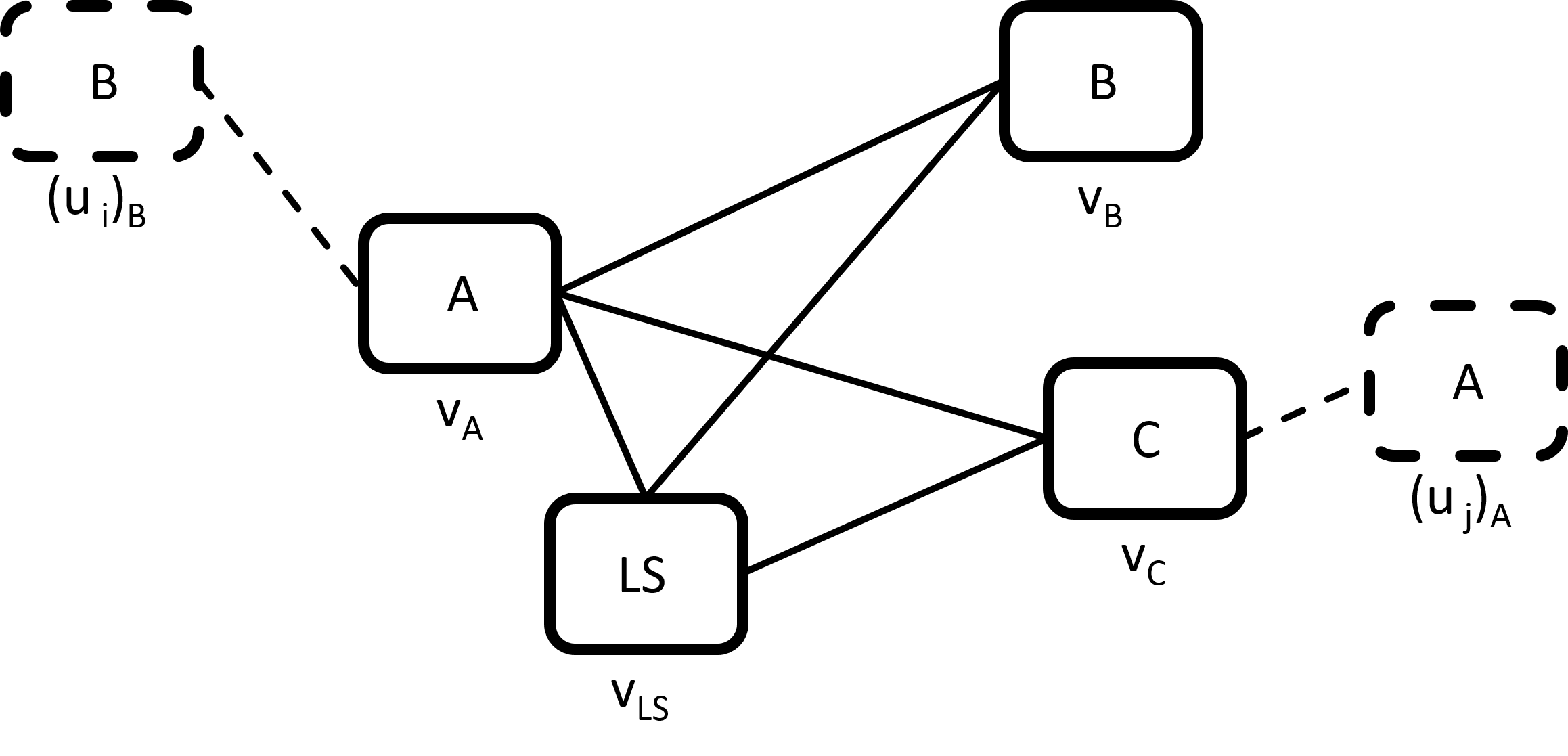}
\end{center}
\caption{The network `gadget' associated with vertex $v$ in the reduction from the Hamiltonian Path problem to network reachaility. The vertex $v$ has an incoming edge from $u_i$ and an outgoing edge to vertex $u_j$ in the input graph $G$.}\label{Fi:widget}
\end{figure}

\begin{figure}
\begin{center}
\includegraphics[width=0.5\textwidth]{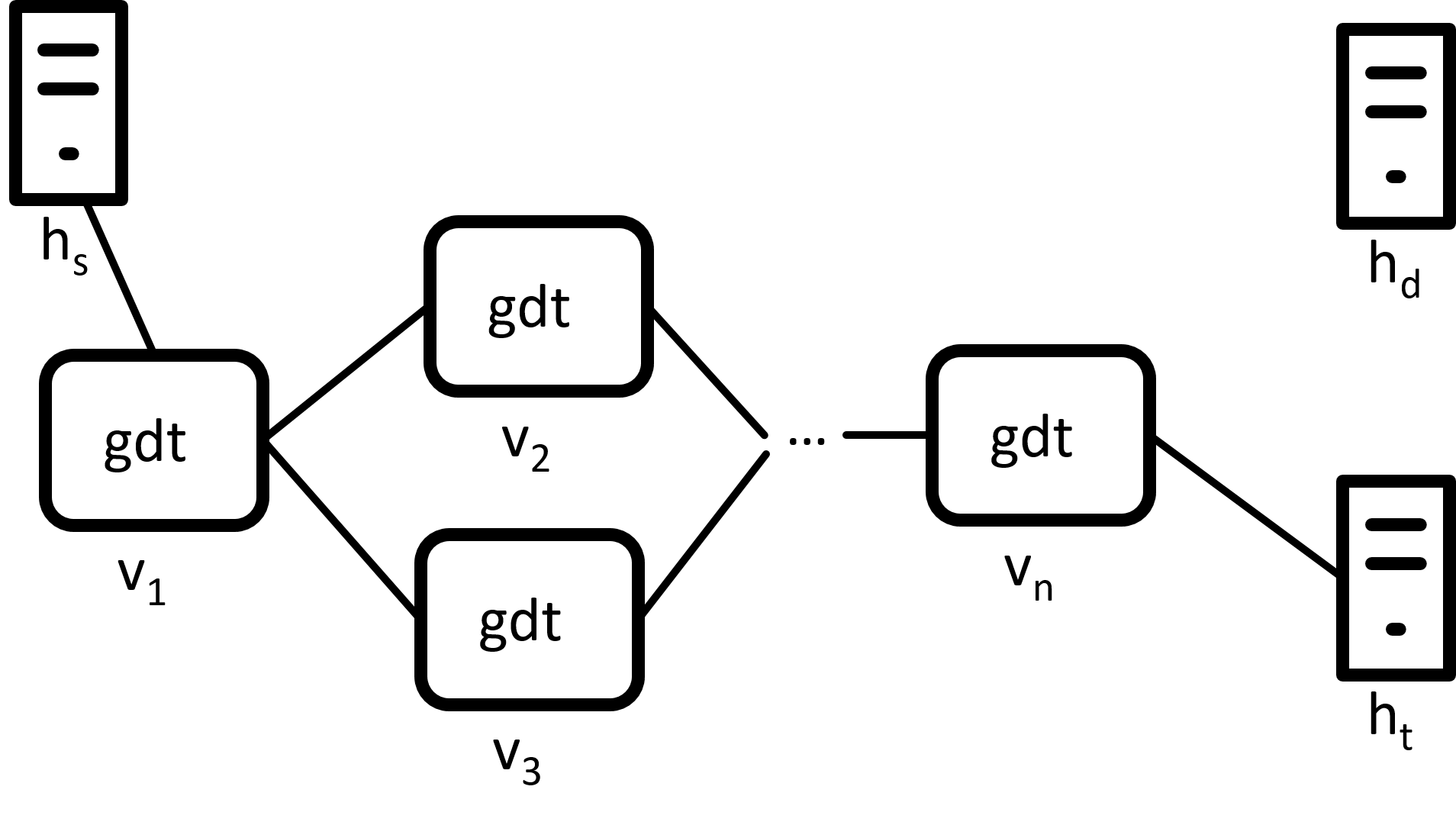}
\end{center}
\caption{The network resulting from the reduction from the Hamiltonian Path problem to network reachability.}\label{Fi:hamiltonian-path-red}
\end{figure}

\subsection{Unordered Safety in arbitrary networks is EXPSPACE-hard.} \label{sec:unordered-expspace-lower}
The result in this section is similar to previous work on message passing
systems with unordered communication
channels~\cite{lipton1976reachability,sen2006model}, and is included here for
completeness of presentation.

The lower bound is obtained by a reduction from the \emph{VASS control state reachability problem}. We first present the problem and its known complexity results.
A \emph{vector addition system with states (VASS)} is a weighted directed graph $(V,E,v_0,w:E\to\Z^k)$, where $V$ is a finite set of vertices (\emph{Control States}), $E\subseteq V\times V$ is a set of directed edges, $v_0$ is the initial vertex, and $w$ is a weight function that assigns a $k$-dimensional weight vector to every edge.
A (finite) path $\pi$ in the directed graph is \emph{valid} if it begins in $v_0$ and every prefix of $\pi$ has a non-negative sum of weights in every dimension.

The \emph{VASS control state reachability problem} gets as input a VASS and a \emph{reachability set} $R\subseteq V$, and checks whether there exists a valid path in the VASS to (at least) one vertex in $R$.

\begin{lemma}[\cite{cardoza1976exponential,lipton1976reachability,rackoff1978covering}]\label{fact:VASSComplexity}
The VASS control state reachability problem is EXPSPACE-complete.
Moreover, it is EXPSPACE-hard even when the coefficients of every vector in the image of the weight function are bounded by $\pm 1$, and even when every vector has at most one non-zero dimension.
\end{lemma}

To simplify our proofs we define the class of \emph{simple VASSs} as all VASSs that satisfy:
\begin{itemize}
\item Every weight vector has exactly one non-zero coefficient which is either $+1$ or $-1$.
\item All the outgoing edges of every vertex $v$ have different weight vectors.
Formally, for every $v_1,v_2,v_3\in V$, if $(v_1,v_2),(v_1,v_3) \in E$ and $w(v_1,v_2)=w(v_1,v_3)$, then $v_2 = v_3$.
\end{itemize}

The next claim is a simple corollary of Lemma~\ref{fact:VASSComplexity}.
\begin{corollary}\label{cor:SimpleVASS}
The control state reachability problem over simple VASS systems is EXPSPACE-hard.
\end{corollary}

Next, we show a reduction from control state reachability over simple VASS systems to stateful network reachability.

The reduction is straightforward:
given a VASS system $(V,E,v_0,w:E\to\Z^k)$ and a reachability set $R\subseteq V$ we construct a network with two hosts, namely $h_1$ and $h_2$ and one middlebox $m$ (see Figure~\ref{fig:ReductionFromVASS}).
The network reachability problem is to determine whether a packet with source host $h_1$ can reach $h_2$.
The set of packet tags is $T=\{1,\dots,k\}$ (where $k$ is the number of dimensions in the VASS system).
We denote by $p_t=(h_1,h_2,t)$, and $P_T=\{p_t \mid t\in T\}$ the packets host $h_1$ sends.
We associate each packet $p_t$ with a vector $\vec{t}\in \N^k$ that consists of $1$ in dimension $t$ and the rest of the dimensions are zero.
The set of states of $m$ is $V$ (with initial state $v_0$) with the addition of one sink state.
When in sink state, the middlebox discards all incoming packets and remains in sink state.
We now describe the transitions of the middlebox $m$ from state $v \in V$:
\begin{itemize}
\item Upon receipt of a packet $p_t$ from port 1:
  \begin{itemize}
  \item If $v\in R$, then forward the packet to port 3 (reachability is obtained).
  \item If there exists $u\in V$ such that $(v,u)\in E$ (of the VASS) and $w(v,u) = \vec{t}$, then:
    \begin{itemize}
    \item Forward $p_t$ to port 2
    \item Change state to $u$
    \end{itemize}
  \item Else (such $u$ does not exists), discard packet and go to sink state.
  \end{itemize}
\item Upon receipt of a packet $p_t$ from port 2:
  \begin{itemize}
  \item If $v\in R$, then forward the packet to port 3 (reachability is obtained).
  \item If there exists $u\in V$ such that $(v,u)\in E$ (of the VASS) and $w(v,u) = -\vec{t}$, then:
    \begin{itemize}
    \item Discard the packet
    \item Change state to $u$
    \end{itemize}
  \end{itemize}
\item Upon receipt of a packet from port 3, go to sink state.
\item Upon receipt of a packet $p \not \in P_T$ from any port, go to sink state.
\end{itemize}

In order to prove the correctness of the reduction we give the next definitions and notations.
A \emph{VASS configuration} is a tuple $(v,\vec{c})\in V\times \N^k$ which consists of a vertex and a vector.
A configuration is reachable in $n$ steps if there exists a valid path in the VASS with length exactly $n$ and total sum of weights $\vec{c}$.
We denote by $S_\mathit{VASS}(n)$ the (finite) set of all configurations that are reachable in $n$ steps.

A \emph{VASS-network configuration} is a tuple $(v,\vec{c})\in V\times \N^k$, where $v$ is the state of the middlebox $m$ and $\vec{c}$ corresponds to the multiplicity of the packets of $P_T$ in the multiset of packets in port 2.
That is, if the multiplicity of packet $p_t$ in the multiset is $r$, then dimension $t$ of $\vec{c}$ is $r$.
We say that a VASS-network configuration is reachable in $n$ steps if there exists a scenario that consists of exactly $n$ middlebox packet processing events that forms the configuration.
We denote by $S_\mathit{Network}(n)$ the (finite) set of all VASS-network configurations that are reachable in $n$ steps.

\begin{lemma}\label{lem:ProofOfReduction}
For every $n \geq 0$: $S_\mathit{VASS}(n) = S_\mathit{Network}(n) - (\{\mathit{sink}\}\times \N^k)$.
\end{lemma}
\begin{proof}
The proof is by induction over $n$. The proof for $n=0$ is trivial.
For $n>0$, let $(v,\vec{c})$ be an arbitrary VASS configuration in $S_\mathit{VASS}(n-1)$. We claim that every successor configuration of $(v,\vec{c})$ is also in $S_\mathit{Network}(n)$.
The proof is straightforward. If the successor is reachable by an addition of positive vector $\vec{r}$, then a corresponding successor in the network is obtained when $h_1$ sends a packet of type $r$ and $m$ processes the packet.
If the successor is reachable by an addition of negative vector $\vec{r}$, then by the induction hypothesis there exists a pending packet in port 2 with type $r$, and a successor in the network is obtained when $m$ processes one packet from port 2 with type $r$.
Hence, we get that $S_\mathit{VASS}(n) \subseteq S_\mathit{Network}(n) - (\{\mathit{sink}\}\times \N^k)$.
The proof that $S_\mathit{Network}(n) - (\{\mathit{sink}\}\times \N^k) \subseteq S_\mathit{VASS}(n)$ follows from similar arguments.
\qed
\end{proof}

The next lemma follows immediately from Lemma~\ref{lem:ProofOfReduction} and Corollary~\ref{cor:SimpleVASS}.
\begin{lemma}\label{cor:EXPSPACEhard}
The reachability problem w.r.t. the unordered network semantics for an arbitrary network is EXPSPACE-hard.
\end{lemma}

\begin{figure}
\begin{center}
\includegraphics[width=0.5\textwidth]{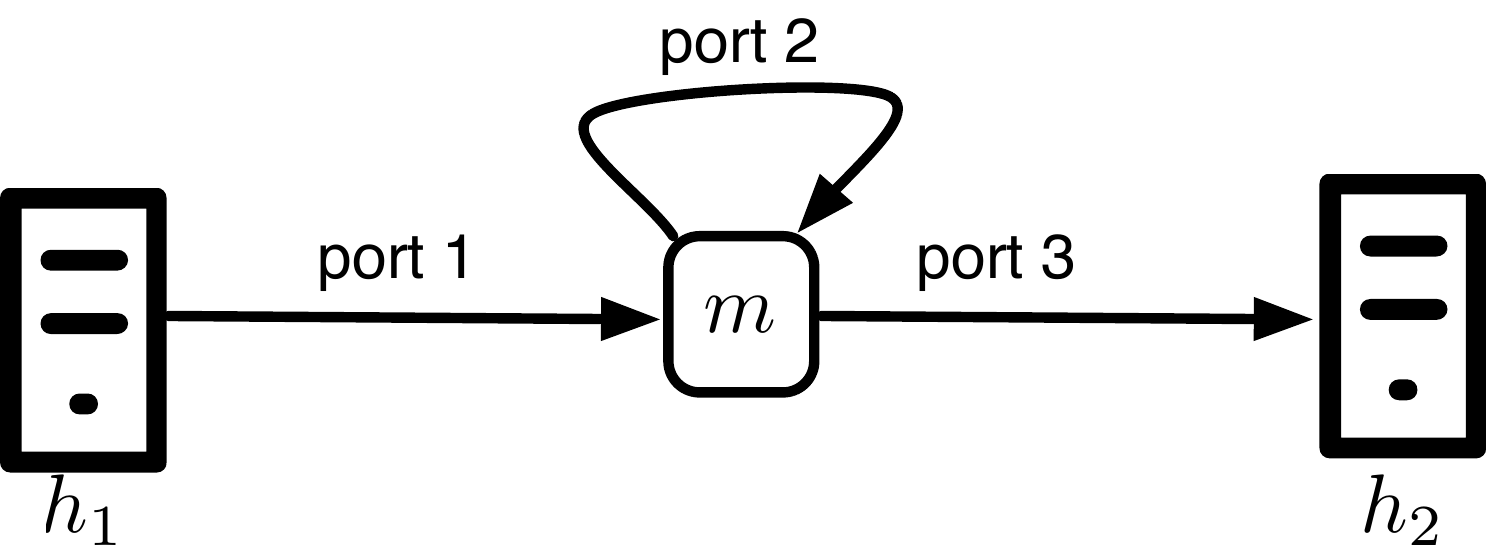}
\end{center}
\caption{The network resulting in the reduction from the VASS control state reachability problem.}\label{fig:ReductionFromVASS}
\end{figure} 
\section{Upper Bounds on Complexity of Safety w.r.t. the Unordered Semantics} \label{sec:complexityUpperBounds}

This section provides complexity upper bounds for the safety problem of stateful networks w.r.t. the unordered semantics of networks.
Our complexity analysis considers symbolic representations of middleboxes (which might be exponentially more succinct than explicit-state representations).
The obtained upper bounds match the lower bounds from \secref{complexityLowerBounds} (hence, the bounds are tight).

\begin{remark}\label{rem:SmallArity}
The complexity upper bounds we present are under the assumption that all relations used to define middlebox states may have at most polynomial number of elements (polynomial in the size of the network and the size of the middlebox representation).
To enforce this limitation we assume that the arity of relations is constant.
If the arity of the relation is bounded by a constant $c$, then the number of elements is bounded by the polynomial $n^c$, where $n$ is the size of the network.

In all of our examples we use relations with arity at most three, and since abstract packets have only three attributes, we believe that most applications will use relations with small arity.

\end{remark}

\paragraph{The Input to the Safety Verification Problem}
The input to the safety verification problem is given in the form of a network topology description, and the symbolic representations of the middleboxes in the network.

The complexity results in this section are given in terms of the number of hosts in the network $|H|$, the size of the type domain $|T|$, the total number of ports in the network $|\PortSet|$, the number of middleboxes in the network $|M|$, and the total size of the symbolic representation $|S|=\sum |S_i|$ where $|S_i|$ is the size of the symbolic representation of middlebox $m_i$.

In our complexity analysis we sometime refer to the set of packets in the networks. Recall that the set of packets in the networks is $P=H\times H\times T$, and so the size of $P$ is $|P|=|H|^2|T|$.
Finally, in our complexity analysis we also refer to $\sum |\mathcal{R}_i|$ which denotes the total size of the domains of relations of middleboxes in the network where $\mathcal{R}_i$ is the domain of relation $R_i$.
Note that $|\mathcal{R}_i|$ is polynomial in the in the size of $|H|$, $|\PortSet|$ and $|T|$, as the arity of $R_i$ is fixed and the domains of its dimensions are taken from $H$, $\PortSet$ and $T$.

\subsection{Unordered Safety of Increasing Networks is in PTIME} \label{sec:upper-increasing}
In this section, we show that safety of syntactically increasing networks is in PTIME.
\begin{figure}[t]
\begin{tabbing}
mm\=mm\=mm\=mm\=mm\=\kill
\> $\mathit{StateData}$ := $\{m \mapsto \mathit{InitialRelationValues}(m) \mid m \in M\}$\\
\> $\mathit{PacketData}$ := $\{m \mapsto \mathit{NeighborHostPackets}(m) \mid m \in M\}$\\
\> $\bwhile$ fixed-point not reached \\
\>\>$\bforeach~m\in M$, $(p,\SinglePort) \in \mathit{PacketData}(m)$\\
\>\>\>let $q = \mathit{StateData}(m)$\\
\>\>\>$\bif$ $\delta_m(q,(p,\SinglePort)) = \emptyset$  $\bthen$ $\breturn$ violation //  abort state reached \\
\>\>\>let $(q',o) \in \delta_m(q,(p,\SinglePort))$\\
\>\>\>$\mathit{StateData}$ := $\mathit{AddData}(m,q')$\\
\>\>\>$\mathit{PacketData}$ := $\mathit{AddPacketsToNeighbors}(m,o)$\\
\>$\breturn$ safe
\end{tabbing}
\caption{\label{Fi:ptime}%
Safety checking of increasing networks.}
\end{figure}

\figref{ptime} presents a polynomial algorithm for determining safety of a syntactically increasing network. The algorithm performs a fixed-point computation of the set of all tuples present in middlebox relations in reachable middlebox states,
as well as the set of all different packets transmitted in the network.
For every middlebox $m \in M$, the algorithm maintains the following sets:
\begin{itemize}
\item $\mathit{StateData}(m)$: a set of pairs of the form $(R,\overline{d})$ where $R$ is a relation of $m$, and $\overline{d}$ is a tuple in the domain of $R$, indicating that there is a run in which $\overline{d}\in R$.

\item $\mathit{PacketData}(m)$: a set of pairs of the form $(p,\SinglePort)$, where $p$ is a packet and $\SinglePort$ is a port of $m$, indicating  that $p$ can reach $m$ from port $\SinglePort$.

\end{itemize}
$\mathit{StateData}(m)$ is initialized to reflect the initial values of all middlebox relations. $\mathit{PacketData}(m)$ is initialized to include the packets $P_h$ that can be sent from neighbor hosts $h \in H$.
As long as a fixed-point is not reached, the algorithm iterates over all middleboxes and their packet data. For each middlebox $m$ and $(p,\SinglePort) \in \mathit{PacketData}(m)$, $m$ is run over $(p,\SinglePort)$ from a state $q$ in which every relation $R$ contains all the tuples $\overline{d}$ such that $(R,\overline{d}) \in \mathit{StateData}(m)$. The sets $\mathit{StateData}(m)$ and $\mathit{PacketData}(m')$ for every neighbor $m'$ of $m$, are updated to reflect the discovery of more elements in the relations (more reachable states), and more packets that can be transmitted.

As the algorithm only adds relation elements and packets, the number of additions is bounded by $(|P||\PortSet|+\sum |\mathcal{R}_i|)$. At every iteration of the $\bwhile$ loop, at least one relation element or packet is added to $\mathit{StateData}$ or $\mathit{PacketData}$ respectively.
The number of $\bforeach$ iterations in every single $\bwhile$ iteration is bounded by $|P||\PortSet|$.
The runtime of every $\bforeach$ iteration is linear in the runtime of the corresponding middlebox, which is linear in the size of its symbolic representation.
This is because the computation of $\delta_m(q,(p,\SinglePort))$ consists of executing the middlebox program, and since the symbolic representation does not have loops, the runtime is linear. Hence, the runtime of a single iteration of the $\bforeach$ loop can be bounded by $|S|$.

The total running time of the algorithm is then bounded by $(|P||\PortSet|+\sum |\mathcal{R}_i|)|P||\PortSet||S|$, and hence polynomial.

The correctness of the algorithm relies on the next lemma, which is a variation of Lemma~\ref{lem:PacketCanComeAgain}.
\begin{lemma}\label{lem:PacketCanComeAgainAlsoUnordered}
For every increasing network,
if there is a run in the unordered semantics in which packet $p$ arrives to port $\SinglePort$ of middlebox $m$, then any run $r$ in the unordered semantics has an extension in which packet $p$ arrives to $m$ from port $\SinglePort$.
Moreover, if there is a run in which element $\overline{d}$ is in a relation $R$, then any run has an extension in which element $\overline{d}$ is in the relation $R$.
\end{lemma}

We now use Lemma~\ref{lem:PacketCanComeAgainAlsoUnordered} to prove that in every iteration the data structure of the algorithm under-approximates $\mathit{PacketData}$ and $\mathit{StateData}$.

\begin{lemma}\label{lem:UnderApprox}
For every iteration of the algorithm there is a run $r$, such that
if $(p,\SinglePort)\in\mathit{PacketData}(m)$, then in $r$ there is a step in which $p$ arrived to $m$ from port $\SinglePort$, and if $(R,\overline{d})\in\mathit{StateData}(m)$, then in $r$ there is a step in which $\overline{d}$ was added to $R$.
\end{lemma}

\begin{proof}
The proof is by induction on the number of iterations performed by the algorithm. The proof for the base case (zero iterations performed) is trivial --- the initial state of the $\mathit{PacketData}$ and $\mathit{StateData}$ matches the initial state of the network.

For the $n$-th iteration, let $(p,\SinglePort)\in\mathit{PacketData}(m)$. We consider two distinct cases. In the first case, after the $n-1$-th iteration, $(p,\SinglePort)\in\mathit{PacketData}(m)$. Then by the induction hypothesis, there exists a run $r$ such that in $r$ there is a step in which $p$ arrived to $m$ from port $\SinglePort$.
In the second case, $(p,\SinglePort)$ was added to $\mathit{PacketData}$ in the $n$-th iteration. In this case, after iteration $n-1$ there must have existed a middlebox $m'$ adjacent to $m$, a state $q$ in which $\{(R_1,\overline{d_1}),\cdots,(R_k,\overline{d_l})\}\subseteq \mathit{StateData(m')}$, and $(p',\SinglePort')$, such that as a result of running $m'$ over $(p',\SinglePort')$ from state $q$, $(p,\SinglePort)$ was sent to $m$.
By the induction hypothesis, there exist runs $r_{1,1},\cdots,r_{k,l}$ in which $(R_1,\overline{d_1}),\cdots,(R_k,\overline{d_l})$ (respectively) are added to $\mathit{StateData(m')}$, as well as a run $r_0$ in which $p'$ arrives to $m'$ from $\SinglePort'$.
Then by \lemref{PacketCanComeAgainAlsoUnordered} we can constructs a run $r'$ in which $m'$ is in state $q$ and $p'$ has arrived to $m'$ from $\SinglePort'$.
The configuration $c$, which is obtained by $m'$ processing $p'$, is a successor of the last configuration of $r'$. We denote the resulting run by $r$, and note that in the last step of $r$, $p$ arrived to $m$ from port $\SinglePort$.

The proof for $(R,\overline{d})\in\mathit{StateData}(m)$ follows from similar arguments.

Finally we use \lemref{PacketCanComeAgainAlsoUnordered} to construct a witness run for the n-th iteration.
\qed
\end{proof}

The next lemma shows that when fixed-point occurs the data structure over-approximate $\mathit{PacketData}$ and $\mathit{StateData}$.

\begin{lemma}\label{lem:OverApprox}
When the algorithm reaches a fixed-point, if $(p,\SinglePort)\notin\mathit{PacketData}(m)$ (respectively., $(R,\overline{d})\notin\mathit{StateData}$), then there is no run in which $m$ receives $p$ from port $\SinglePort$ (resp., $\overline{d}$ is added to $R$).
\end{lemma}

\begin{proof}
Let $r$ be the witness run that the fixed-point under-approximates ($r$ exists by Lemma~\ref{lem:UnderApprox}).
Towards a contradiction we assume that there is a run $r'$ in which $m$ receives $p$ from port $\SinglePort$ (respectively, $\overline{d}$ was added to $R$), but such an event did not occur in $r$.
By Lemma~\ref{lem:PacketCanComeAgainAlsoUnordered}, we get that $r$ has an extension in which the event does happen.
But such an extension contradicts the fact that a fixed-point occurred.
Hence, the data structure over-approximates all runs.
\qed
\end{proof}
Lemma~\ref{lem:UnderApprox} and Lemma~\ref{lem:OverApprox} imply that the algorithm determines the safety problem, and the next theorem follows.

\begin{theorem}
The safety problem of syntactically increasing networks w.r.t. the unordered semantics is in PTIME.
\end{theorem}

\begin{proof}
Safety is violated iff there exists a run $r$ that ends in a configuration $c$ where some middlebox is in state $q$ with packet $p$ pending on its port $\SinglePort$ such that $\delta_m(q,(p,\SinglePort))=\emptyset$.

By lemmas~\ref{lem:UnderApprox} and \ref{lem:OverApprox}, the latter holds iff at some iteration of the algorithm $(p,\SinglePort)\in PacketData(m)$, and the values pf $m$'s relations in state $q$ are included in $StateData(m)$, in which case the algorithm identifies the safety violation.
\qed
\end{proof}

\begin{remark}
Recall that for increasing networks, safety w.r.t. the unordered semantics and the FIFO semantics coincide. As such, the polynomial upper bound applies to both.
\end{remark}

\begin{remark}\label{rem:nTagsIsPSPACE}
The complexity analysis of the algorithm used the property that $|P|$ is polynomial in the network representation.
If $n$-tag packet headers are allowed, i.e. $P = H \times H \times T_1 \ldots \times T_n$, then $|P|$ is no longer polynomial in the network representation, damaging the complexity analysis of the algorithm.
In fact, in this case the safety problem w.r.t. the unordered semantics becomes PSPACE-hard even for stateless middleboxes.

Intuitively, $n$-tag packet headers allow a middlebox to maintain the state of $n$ automata  in the packet header, supporting a reduction from the emptiness problem of the intersection of $n$ automata, which is PSPACE-hard~\cite{kozen1977lower}.
\end{remark}

\begin{proof}
The PSPACE-hardness proof is by reduction from the problem of deciding the emptiness of intersection of $n$ automata~\cite{kozen1977lower}, which is formally defined as:
\begin{itemize}
\item Input: $n$ automata $A_1,\dots,A_n$ over alphabet $\{0,1\}$ with state set $Q$ (w.l.o.g. all automata have the same set of states).
\item Question: is there a word $w\in \{0,1\}^*$ that is accepted by all $n$ automata?
\end{itemize}
The reduction is as follows.
Given $n$ automata with state set $Q$ we define a network with one host and one middlebox.
The packets consist of $n+1$-tuples of tags from the domain $T = Q \cup \{0,1\}$.
Intuitively, the first $n$ tags hold the states of the $n$ automata, and the last tag is an input symbol for the automata.
The middlebox has two ports.
Port $0$ is connected to the host and port $1$ is a self loop.

The symbolic representation of the middlebox has four parts:
\begin{enumerate}
\item \emph{Initial state verifier.} The first part handles packets from port $0$. If the packet's first $n$ tags do not correspond to the $n$ initial states, then the middlebox discards the packet.
Otherwise it sends the packet to port $1$.
\item \emph{Advance state.}
The second part handles packets from port $1$.
In a sequence of $n|Q|$ commands, the program advances the state of each automaton (i.e., changes the corresponding packet tag) according to the symbol in tag $n+1$.
After the sequence, the program continues to the third part.
\item \emph{Accepting state verifier.}
If the packet's tag corresponds to $n$ accepting states, then the program aborts.
Otherwise the program continues to the fourth part.
\item \emph{New symbol generator.}
In the fourth part the program generates two packets that differ only in their $n+1$ tag.
In one packet the tag has value $0$ and in the second it has value $1$.
Both packets are sent back to port $1$.
\end{enumerate}

It is an easy observation that the intersection of the $n$ automata is non-empty iff abort is invoked.
\qed
\end{proof}

\subsection{Unordered Safety of Progressing Networks is in coNP}\label{sec:AcyclicInNP}
We prove coNP-membership of the safety problem in syntactically progressing networks by proving that there exists a witness run for safety violation if and only if there exists a ``short'' witness run, where a witness run for safety violation is a run from the initial configuration  in which at least one middlebox reaches an abort state.

The proof considers the \emph{network states} that arise in a run.
A \emph{network state} captures the states of all middleboxes (not to be confused with a network configuration, which also includes the content of every channel). Formally, let $\Net$ be a network whose middleboxes are defined symbolically via (in total) $n$ relations, namely $R_1,\dots,R_n$. Then the \emph{network state} consists of the values of $(R_1,\dots,R_n)$.

In order to construct a ``short'' witness run, we wish to bound both the number of different network states in a run and the number of steps in which a run stays in the same state. The former is bounded due to the progress of the network: once the state of some middlebox changes along a run, it will not change back to the previous state.
The latter is more tricky. To provide a bound, we wish to analyze the packets that ``affect'' the run. We define the notion of \emph{active packets}. The active packets are a superset of the packets that actually affect the run.

\Heading{Active packets}
Let $r$ be a finite run of a network. We say that a packet $p$ is \emph{active} in step $i$ of $r$, if it resides in the ingress channel of some middlebox $m$ and it is processed (i.e., received by $m$) in some future step of $r$.
A packet is \emph{inactive}, if it is pending in the ingress channel of $m$ until the end of the run.

The next lemmas show that only a few active packets are needed to reach a certain state in the network.
Intuitively, the proof of the lemma traverses the run from the last configuration to the first, and removes inactive packets (and steps that produce only inactive packets), which in turn makes other, earlier, packets inactive. For a run $r$ and a network state $s$ that appears in $r$, we denote by $r[s]$ an interval of the run that includes all consecutive occurrences of $s$ (for runs of progressing networks, the interval is unique).

\begin{lemma}\label{lem:UnfreezePackets}
Let $r$ be a run in which the network state changes exactly $k$ times, and the different states are $s_1,s_2,\dots,s_k$ (in this order).
Then for every prefix $r_{s_i}$ of $r$ that ends in a state $s_i$, there is an extension $e_{s_i}$ to $r_{s_i}$ such that:
(i)~$e_{s_i}$ visits the network states $s_i,\dots,s_k$;
(ii)~$e_{s_i}$ has at most $k-i$ active packets in every step; and
(iii)~the number of active packets in $e_{s_i}$ may decrease only after a change in the network state.
\end{lemma}

\begin{proof}
The proof is by induction over $|r| - |r_{s_i}|$.
For the base case $r = r_{s_i}$ and the proof is trivial.
For $|r| > |r_{s_i}|$, we extend the prefix $r_{s_i}$ by one step according to $r$.
We denote this extended prefix by $r'$. Let $p$ be the last packet that was processed in $r'$, and let $m$ be the middlebox that processes $p$. That is, $m$ and $p$ are responsible for the step that extends $r_{s_i}$ to $r'$.

We consider two distinct cases.
In the first case, the network state in the last configuration of $r'$ is still $s_i$.
Then by the induction hypothesis we get that there is an extension $e'_{s_i}$ with at most $k-i$ active packets in interval $e'_{s_i}[s_i]$.
We consider the set of packets that were created by $m$ after processing $p$.
If this set has at least one active packet in $e'_{s_i}$, then we define $e_{s_i}$ to be $e'_{s_i}$ prepended by the last step of $r'$, where $p$ is marked as active and all the active packets of $e'_{s_i}$ remain active. Surely, there are no more than $k-i$ active packets in the first step of $e_{s_i}$ since at least one of the active packets in $e'_{s_i}$ resulted from $p$ and hence did not yet exist in this step, so it balances out the addition of $p$ as an active packet. In addition, the total number of active packets is not decreased in this step (thus, the claim holds).
Otherwise, we define $e_{s_i}$ to be $e'_{s_i}$, i.e. we skip the processing of $p$, and turn it to inactive.

In the second case, the last state in $r'$ is $s_{i+1}$.
Then by the induction hypothesis we get that there is an extension $e'_{s_{i+1}}$ with at most $k-i-1$ active packets.
In this case we construct $e_{s_{i}}$ simply by prepending to $e'_{s_i}$ the last step of $r'$. That is, $p$ is marked as active and all the active packets of $e'_{s_{i+1}}$ remain active.
There are only $k-i-1 + 1 = k-i$ active packets.
Hence, the claim holds.
This completes the proof.
\qed
\end{proof}

\begin{lemma}\label{lem:UnfreezePacketsFrequentChange}
Let $r$ be a run in which the network state changes exactly $k$ times, and the different states are $s_1,s_2,\dots,s_k$ (in this order).
Then there exists a run $r'$ such that:
(i)~$r'$ visits the network states $s_1,s_2,\dots,s_k$; and
(ii)~$r'$ stays in state $s_i$ at most $(k-i)^2|P||M|$ steps.
\end{lemma}

\begin{proof}
For the sake of the proof we give a unique id to every active packet according to the following rules:
\begin{itemize}
\item If a host sends an active packet, then the packet gets some unique id (for example, maximal id assigned so far + 1).
\item If an active packet $p_1$ was processed by a middlebox, and the middlebox forwards only one active packet $p_2$, then $p_2$ gets the id of $p_1$.
\item If an active packet $p_1$ was processed by a middlebox, and the middlebox forwards more than one active packet, then each active packet gets a unique id (for example, maximal id assigned so far + 1).
\end{itemize}
We now return to the proof. Let $e'$ be the shortest extension for the prefix of $r$ that consists of the initial configuration that satisfies the assertions of Lemma~\ref{lem:UnfreezePackets}.
The extension $e'$ clearly visits $s_1,\dots,s_k$. We claim that it stays in state $s_i$ at most $(k-i)^2|P||M|$ steps.
The proof of the claim follows from the fact that if there are two steps $j_1 < j_2$ in $e'[s_i]$ such that in both steps a middlebox $m$ received an active packet $p$ with id $\mathit{id}$, and no new active packet (i.e., an active packet with a new packet id) was generated between those rounds, then a run in which $m$ does not process packet $p$ with id $\mathit{id}$ is shorter by one step, and reaches the same configuration in step $j_2 - 1$.
Hence, if a certain middlebox processed more than $|P|(k-i)$ packets, then it must be the case that either a new active packet was created, or it processed the same packet twice.
The proof is complete by the pigeonhole principle and by the fact that there are at most $k-i$ active packets and $|M|$ middleboxes.
\qed
\end{proof}

The next lemma shows that there is a short witness for reachability of a state in progressing networks.
\begin{lemma}\label{lem:AcyclicHasShortWitness}
Let $\Net$ be a syntactically progressing network whose middleboxes are defined symbolically via relations $R_1,\dots,R_n$ (in total).
Then there is a run ending in an abort state if and only if there is such a run whose length is at most $(\sum_{i=1}^n |\mathcal{R}_i|)^3|P||M|$.
\end{lemma}
\begin{proof}
The proof is an immediate corollary of Lemma~\ref{lem:UnfreezePacketsFrequentChange}.
If there is a run $r$ that leads to a certain state of $R_1,\dots,R_n$, then since all middleboxes are progressing we get that the number of intermediate network states $k$ is at most $(\sum_{i=1}^n |\mathcal{R}_i|)$.
We denote the intermediate states by $s_1,\dots,s_k$. By Lemma~\ref{lem:UnfreezePacketsFrequentChange}, there is also a run $r'$ that visits the same $k$ states and stays in state $s_i$ at most $(k-i)^2|P||M| \leq k^2|P||M|$ steps. Therefore $|r'| \leq k^3|P||M|$.
\qed
\end{proof}

Since the size of each relation is polynomial in the size of the network, we conclude:
\begin{theorem}\label{thm:AcyclicIsNPComplete}
The safety problem w.r.t. the unordered semantics for progressing networks is coNP-complete.
\end{theorem}

\begin{proof}
The lower bound follows from Lemma~\ref{lem:progressingIsNPHard}.
The upper bound is obtained by first observing that the complement of the safety problem is polynomially reducible to the reachability of a state in the network (by adding a special abort state).
In addition, the state reachability problem is in NP: since the arity of each relation in the considered middlebox programs is fixed, its size is polynomial in the size of the network.
Hence, by Lemma~\ref{lem:AcyclicHasShortWitness}, there is a witness run for reachability whose length is polynomial.
Thus, the NP procedure is to guess the short run and verify it, in time linear in the length of the run multiplied by $|S|$ (the size of the symbolic representation of the middleboxes which also bounds the time it takes to compute their transitions).
\qed
\end{proof}

\subsection{\label{sec:expspaceUpper}Unordered Safety of Arbitrary Networks is in EXPSPACE}

In this section we show how to solve the reachability problem of symbolic networks by a reduction to the \emph{coverability problem} of \emph{Petri Nets}, which is EXPSPACE-complete~\cite{lipton1976reachability,rackoff1978covering}.

Similarly to the lower bound result (\secref{unordered-expspace-lower}), the
upper bound result on the complexity of safety of arbitrary networks is
similar to previous work (\cite{lipton1976reachability,sen2006model}), and is
included here for completeness of presentation.

A Petri Net is a four-tuple $\mathcal{C} = (\mathcal{P}, \mathcal{T}, \mathcal{I}, \mathcal{O})$ where $\mathcal{P}$ is a set of \emph{places}, $\mathcal{T}$ is a set of \emph{transitions}, $\mathcal{I}: \mathcal{T} \to \Nat^\mathcal{|P|}$ is an \emph{input function} and $\mathcal{O}: \mathcal{T} \to \Nat^\mathcal{|P|}$ is an \emph{output function}. A \emph{marking} $\mu \in \Nat^\mathcal{|P|}$ denotes the number of \emph{tokens} assigned to each place.
Given a marking, a transition $t\in \mathcal{T}$ can be \emph{fired} (equivalently  \emph{enabled}) if $\mathcal{I}(t) \leq \mu$. Firing a transition $t\in \mathcal{T}$ from marking $\mu$ produces a new marking $\mu'=\mu-\mathcal{I}(t) + \mathcal{O}(t)$~\cite{peterson1977petri}. We denote a firing of a transition by $\mu \to_t \mu'$.
In the following, we will refer to non-zero dimensions in $\mathcal{I}(t)$ as \emph{consumed} tokens, and non-zero dimensions in $\mathcal{O}(t)$ as \emph{produced} tokens.
A finite run in a Petri Net from a marking $\mu_0$ is a series of transitions and resulting markings $\mu_0\to_{t_0}\mu_1\to_{t_1}\dots\to_{t_k}\mu_k$ s.t. $t_0$ can be fired from $\mu_0$ and each following transition can be fired from the previous marking.

The coverability problem asks, given a Petri Net $\mathcal{C}$, an initial marking $\mu_0$ and a target marking $\mu$, whether there is a finite run leading to a marking $\mu'$ s.t. $\mu' \geq \mu$.

We now show how we encode a symbolic network as a Petri Net, and how we formulate the reachability problem as a Petri Net coverability problem.
We first describe the role of every place and the initial marking, and then we describe the set of transitions used to simulate a run of the network.

\paragraph{Places}
The places are partitioned to sets of places in the following way:
\begin{itemize}

\item Channel places. To keep track of the packets over the unbounded channels, we assign a place to every pair of packet $p\in P$ and channel. The number of tokens in the place corresponds to the number of instances of packet $p$ on the channel. The initial marking for each packet place is $0$.

\item Active and non-active relation places. For every element $\overline{d}$ in every relation $R$ in every middlebox we have two places. The active place will have the marking $1$ when the element is in the relation. When the element is not in the relation the non-active place will the marking $1$. The initial marking for the active (respectively, non-active) place is $1$ if initially the element is in the relation (resp., not in the relation).
Otherwise, the initial marking is $0$. The markings for both places will only be $0$ or $1$.
We need two places since the Petri Net semantics does not allow to encode negative (i.e., non-membership) conditions.

\item Global command place. We have a single place that is used to make sure that at most one middlebox is processing a packet in every step.
The initial marking for the place is $1$; it is consumed whenever a packet processing starts, and produced when it ends.

\item Command places. We have a place for every triple of command, processed packet and input port in every middlebox in the network. The markings on the places are used to keep track of the next command to be executed. In particular, each guarded command block has a single place  (for every combination of packet and input port) rather than a place for each guarded command in the block. This ensures that only one of the guarded commands in the block whose guards evaluate to true is executed.
Having a separate command place for every packet processed and every input port allows us to evaluate variables that appear in the command (including the guards).
The initial marking for the topmost guarded command block in each middlebox (with every combination of packet and input port) is $1$. The initial marking for the rest is $0$.

\item Auxiliary guard places. To allow conjunction and disjunction in the guard we add auxiliary guard places. The initial marking for each of these places is $0$.

\item Abort place. To keep track of the safety state of the network, we assign a single place for all $\babort$ calls made during the network run. The initial marking for the place is $0$.

\end{itemize}

\paragraph{Transitions}
For each middlebox in the network we define a ``command transition'' for each combination of processed command, input packet, input port, and next command, as explained below.
For some commands only a single ``next'' command exists, however, since we allow non-determinism, some commands (specifically, guarded command blocks with overlapping guards) have multiple ``next'' commands, in which case a separate transition is defined for each one of them.

For a guarded command block we define a \emph{set} of ``command transitions''. This allows us to handle complex guards (i.e. guards which contain conjunction and disjunction in addition to atomic propositions). To do so, we recursively decompose each guard while producing a sequence of transitions that simulates the evaluation of the boolean formula in the guard.

To correctly simulate cases in which no guard in a guarded command block is evaluated to \emph{true}, and as a result no command is processed,
we add a \emph{default} guarded command to each guarded command block. The guard of the default guarded command is a conjunction of the negations of the guards of the other guarded commands in the block. The command of the default guarded command is $\bfrw~\emptyset$.

Each of the command transitions of the first command in the middlebox (i.e. the topmost guarded command block) consumes a token from the global command place, and each terminating command that can be executed in the middlebox run produces a token in the global command place. Note that the addition of default guarded commands as described above means that the terminating commands are well defined (i.e. for every command  in the middlebox, if it is terminating in some run then it is a terminating command in every run that it is executed in).
Each of the command transitions of the first command in the middlebox also consumes a token from the corresponding channel place.
Furthermore, every command transition consumes its command place, and produces the command place of the following command, specifically the place corresponding to the combination of the next command to be executed and the same input packet and input port as the packet and port processed in the current command (or the first command in case it is a terminating command).

In addition to the above, the command transition associated with a command, input packet, input port and next command consumes and produces tokens in the places relevant to the corresponding command, as well as the guards (in the case of a guarded command block), as described below.

Since we have a command transition for every combination of command, input packet and input port, when we translate the command to a transition we consider the values of the variables ($src$, $dst$, $type$ and $port$) at that transition based on the packet and port currently processed by the middlebox, and simplify the command (and guards) accordingly. For example, for the command $\ttrusted.\binsert~dst$, packet $(h_0,h_1,t_0)$ and port $\SinglePort_0$, the command simplifies to $\ttrusted.\binsert~h_1$.
In particular, atomic equality predicates are now essentially equalities between constants, and are trivially simplified.

The transition for each guarded command in a guarded command block consumes a token from the command place for the guarded command block, and produces a token in the command place of the first command in the guarded command, as well as consuming and producing the tokens of the guard as described below.

We begin by describing the tokens consumed and produced by the atomic propositions of the guards (after simplification).
Note that since guards do not change the state of the network, all tokens consumed by the guard must also be produced by the guard.
\begin{itemize}
\item Relation membership ($\overline{d} \in R$). Consume (and produce) tokens in the active place for element $\overline{d}$  in relation $R$.
\item Negated relation membership ($\overline{d} \notin R$). Consume (and produce) tokens in the inactive place for element $\overline{d}$  in relation $R$.
\end{itemize}

Next, we describe how disjunction and conjunction are handled:
In the case of a guarded command whose guard's formula $\varphi$ contains a disjunction or conjunction, we produce a series of transitions by recursively decomposing the formula, and producing a set of transitions for every decomposition step. Each decomposition step introduces new auxiliary guard places.
We denote by $c_i \simplies_\varphi c_j$ an intermediate step in the decomposition process where $c_i$ is the place that initiates the evaluation of $\varphi$ and $c_j$ is the place of the next step in the execution. Specifically, initially, $c_i$ is the command place for the guarded command and $c_j$ is the command place of the command.
The recursive decomposition of guard $c_i \simplies_\varphi c_j$ is as follows:
\begin{itemize}
\item Conjunction ($\varphi = \varphi_1 \land \varphi_2$).
We introduce five auxiliary places, denoted $c_1$, $c_2$, $c_3$, $c_4$ and $c_5$, two intermediate steps, and four new transitions.
The first transition consumes one token from $c_i$ and produces two tokens in $c_1$.
The second and third transitions consume one token each from $c_1$ and produce a token in $c_2$ and $c_3$ respectively.
We produce two intermediate steps: $c_2 \simplies_{\varphi_1} c_4$ and $c_3 \simplies_{\varphi_2} c_5$.
Finally, we produce a final transition that consumes one token from both $c_4$ and $c_5$, and produces a token in $c_j$.

\item Disjunction ($\varphi = \varphi_1 \lor \varphi_2$).
We introduce four auxiliary places, denoted $c_1$, $c_2$, $c_3$ and $c_4$, two intermediate steps, and four new transitions.
The first transition consumes a token from $c_i$ and produces a token in $c_1$. Likewise, the second transition consumes a token from $c_i$ and produces a token in $c_2$.
We produce two intermediate steps: $c_1 \simplies_{\varphi_1} c_3$ and $c_2 \simplies_{\varphi_2} c_4$.
The third transition consumes a token from $c_3$ and produces a token in $c_j$. Likewise, the fourth transition consumes a token from $c_4$ and produces a token in $c_j$.
\end{itemize}

The process is performed recursively on $c_i \simplies_{\varphi_1} c_j$ and $c_i \simplies_{\varphi_2} c_j$. The process terminates for $c_i \simplies_\varphi c_j$ once $\varphi$ is an atomic proposition, in which case a single transition is produced, which consumes a token from $c_i$, consumes and produces the tokens for the atomic proposition as described above, and produces a token in $c_j$.

Finally, we describe the dimensions consumed and produced by the commands $\bforward$, $\binsert$, $\bremove$ and $\babort$.
\begin{itemize}

\item $\bforward$.
Produce: the appropriate packets in the egress channel.
We note that in the special case of $\bfrw~\emptyset$ no tokens are produced.

\item $\binsert$.
We replace every $\binsert$ command with a guarded command block consisting of two guarded commands. The first guarded command represents the case where the element is already in the relation, in which case the guard will be a relation membership predicate, and the command will be $\bfrw~\emptyset$.
The second guarded command represents the case where the element is not in the relation. The guard of the command will be a negated relation membership predicate to the guard, and the transition produced from the command will consume and produce the following:

Consume: a token from the appropriate non-active place of the new element.

Produce: a token in the appropriate active place of the new element.

\item $\bremove$. Analogous to $\binsert$.

\item $\babort$.
Produce: a token in the \emph{abort} place.

\end{itemize}
This concludes the description of the command transitions.

Finally, for every host $h$ and every packet $p\in P_h$ we have a ``host transition'' that produces a token in the corresponding ingress channel place of the neighbor middlebox.

\paragraph{From Network Safety to Petri Net Coverability}
Non-safety of the network amounts to a run in the Petri Net where an \emph{abort} place gets a token.
The target marking for the coverability problem is therefore a vector of $0$s, with $1$ in the \emph{abort} place.

As the reduction is polynomial, we get that the stateful network reachability problem is in EXPSPACE.

The reduction, combined with the lower bound implies:
\begin{theorem}\label{thm:EXPSPACEcomplete}
The safety problem of arbitrary stateful networks w.r.t. the unordered semantics is EXPSPACE-complete.
\end{theorem}

\section{\label{sec:benchmark}Implementation and Case Studies}

In this section, we present several examples of networks consisting of stateful middleboxes and their safety properties. 
We describe a prototype implementation of a tool for verification of stateful networks, and describe our initial experience while running the tool on the networks listed in \exref{topologies} and illustrated in \figref{MidTop}.
For the experiments we used a machine equipped with a quad core Intel Core i7-4790 CPU and 32GB of memory, running Ubuntu Linux 14.04.

\subsection{Network Examples}

\paragraph{Load Balancer and IDS}
As an example consider the network shown in \figref{benchmark:lbids}. Here $A$ is a host, $lb$ is a load balancer, which can send a packet received from $A$ to either $r_1$ or $r_2$. Both $r_1$ and $r_2$ are rate limiters, \ie they count and limit the number of packets sent between host pairs. Let us consider a case where the administrator wants to ensure that exactly $8$ packets sent by $A$ can be received by $B$. If the load balancer in this case sends packets from $A$ to both $r_1$ and $r_2$, then this rate limit does not hold.

\paragraph{Firewall and Proxy}
Consider the network in \figref{benchmark:fwproxy}. Here, $c$ is a content addressable cache, which on receiving a packet checks if it has previously seen either server $S_1$ or $S_2$ respond to a packet of the same type; if so it sends back the previously observed response, otherwise it forwards the request to the packets original destination. $f$ is a learning firewall. We want to ensure that $A$ cannot receive data from $S_1$, while $B$ should be able to receive data from both $S_1$ and $S_2$. This is complicated by the fact that $c$'s response is based on the packet type: in the current configuration if $B$ sends a request for type $t$ to server $S_1$ then $A$ can access the response by subsequently sending a request with the same type $t$ addressed to server $S_2$. In general this problem is not solvable without changing the cache to be policy aware.

\paragraph{Multi-Tenant Datacenter}
Consider a multi-tenant datacenter such as Amazon EC2 shown in \figref{benchmark:mdc}. In such datacenters each tenant (customer who purchase VMs from the provider) gets to add rules about their VMs, to the firewall to which their VMs are connected. 
For example in \figref{benchmark:mdc}, each tenant $i$ owns VMs $pub_1^i$ and $pri_1^i$, and programs the rules for firewall $f_i$. Given a set of rules for firewall $f_1$ and $f_2$ we verify that VMs of the same tenant can communicate with each other and that $pri$ VMs of one tenant can send packets to $pub$ VMs of the other.

\subsection{results}

\paragraph{Increasing Middleboxes}
Increasing networks are verified using LogicBlox, a Datalog based database system~\cite{aref2015design}.
The Multi-Tenant Datacenter example is an increasing network. Our tool produced a datalog program with 35 predicates, 153 rules and 29 facts. LogicBlox successfully reached a fixed point in 3s, and proved all required properties. 

\paragraph{Arbitrary Middleboxes}
Progressing and Arbitrary networks are verified using LOLA, a Petri-Net model checker~\cite{schmidt2000lola,TRLola2}. 
In the Load Balancer and Rate Limiter example our tool created a P/T net with 243 places and 663 transitions; it was successfully verified in 30ms. 
In the Firewall and Proxy example our tool produced a P/T net with 530 places and 4447 transitions. LOLA successfully verified the resulting petri-net in 0.2s.
\section{Conclusion and Related Work}\label{sec:Related}
In this work, we investigated the complexity of reasoning about stateful networks.
We developed three algorithms and several lower bounds.
In the future we hope to develop practical verification methods utilizing the results in this work. 
Below we 
survey some of the most closely related work and conclude with open questions and future work.

\subsection{Related Work}

\paragraph{Topology-Independent Verification}
The earliest use of formal verification in networking focused on proving correctness and checking security
properties for protocols~\cite{clarke1998using,ritchey2000using}.
Recent works such FlowLog~\cite{FlowLog} and VeriCon~\cite{VeriCon14} also aim to verify the correctness of a given middlebox implementation w.r.t any possible network topology and configuration, e.g., flow table entries only contain forwarding rules from trusted hosts.

\paragraph{Immutable Topology-Dependent Verification}
Recent efforts in network verification~\cite{MaiKACGK11,nsdi:CaniniVPKR12,nsdi:KVM12,CCR:KhurshidZCG12,Verificare,FMACAD:SNM13,anderson2014netkat,NetKat15}
have focused on verifying network properties by analyzing forwarding tables. Some of these tools including HSA~\cite{nsdi:KCZCMW13},
Libra~\cite{zeng2014libra} and VeriFlow~\cite{CCR:KhurshidZCG12}.
These tools perform near real-time verification of simple properties, but they cannot handle dynamic (mutable) datapaths.

\paragraph{Mutable Topology-Dependent Verification}
SymNet~\cite{stoenescu2013symnet} has suggested the need to extend these mechanisms to handle mutable datapath elements.
In their mechanism the mutable middlebox states are encoded in the packet header.
This technique is only applicable when state is not shared across a flow (\ie the middlebox
can punch holes, but do no more), and will not work for cache servers or learning switches.

The work in~\cite{panda2014verifying} is the most similar to our model.
Their work considers Python-like syntax enriched with uninterpreted functions that model complicated functionality.
However \cite{panda2014verifying} do not define formal network semantic (e.g., FIFO vs ordered channels) and do not give any formal claim on the complexity of the solution.

\paragraph{Channel Systems}
Channel systems, also called Finite State Communicating Machines, are systems of finite state
automata that communicate via asynchronous unbounded FIFO channels~\cite{bochmann1978finite,brand1983communicating}.
They are a natural model for asynchronous communication protocols%
 and, indeed, they form the semantic basis of protocol specification languages such as SDL and Estelle.
Unbounded FIFO channels can simulate unbounded Turing machine tape and therefore all verification problems are undecidable. 
Abdulla and Jonsson~\cite{abdulla1993verifying} 
introduced \emph{lossy channel systems} where messages can be lost  in transit.
In their model the reachability problem is decidable but has a non-primitive lower bound~\cite{schnoebelen2002verifying}.

In this work we use unordered (non-lossy) channels as a different relaxation for channel systems.
The unordered semantics over-approximates the lossy semantics w.r.t. safety, as any violating run w.r.t. the lossy semantics can be simulated by a run w.r.t. the unordered semantics where ``lost'' packets are starved until the violation occurs.

The unordered semantics admits verification procedures with elementary complexity, and turns out to be sufficiently precise for many network protocols in which order is not guaranteed and hence not relied on.

\subsection{Future Work}

\paragraph{Exploration of Network Semantics}
In this work we have outlined two possible network semantics, namely FIFO and Unordered packet processing order. Various other network semantics could be considered, along with their effect on expressibility and complexity results, and the precision loss in safety analysis.
One such network semantics is the \emph{Sticky Channel} semantics, where packets can be added by the sending middlebox and read by the receiving middlebox but cannot be removed. This network semantics corresponds to networks in which middleboxes can arbitrarily retransmit messages.

\paragraph{Modelling Packet Payload}
In this work we have only considered packet headers. However, some middlebox behaviour depends on the content of the packet payload (Intrusion Detection Systems are one such example). A potential approach to bridging this gap could be to model middleboxes using register automata. This would allow us to reason about letters from an infinite alphabet, thus modelling the arbitrary nature of packet payloads, while potentially retaining the decidability of reasoning about such systems.

\paragraph{Liveness}
In this work we have limited ourselves to reasoning about safety properties. However, various liveness and performance properties are just as important when approaching the creation of networks. Reasoning about liveness properties such as guarantees on packet arrival, or performance properties such as load estimates or packet traversal times would require the development of a new model for describing the network semantics and middlebox behaviour. 
In particular, unordered semantics are ill suited for most sorts of reasoning on liveness properties.

\paragraph{Further Aspects of Network Security}
In addition to safety properties that can be expressed by checker middleboxes and liveness properties there are various other network security properties that can be considered when reasoning about networks. Non-interference and information leakage are two examples of security properties which cannot be modeled by our current approach.

\paragraph{Reasoning About Progressing Networks Under the FIFO Semantics}
We've seen that in arbitrary networks reasoning is undecidable under the FIFO semantics but EXPSPACE-complete under the unordered semantics, and that for increasing networks the two semantics coincide. This leaves the question of reasoning about progressing network under the FIFO semantics open. 

\begin{acknowledgements}
    This publication is part of projects that have received funding from the
    European Research Council (ERC) under the European Union's Seventh Framework
    Program (FP7/2007--2013) / ERC grant agreement no. [321174-VSSC], and Horizon
    2020 research and innovation programme (grant agreement No [759102-SVIS]).
    The research was supported in part by Len Blavatnik and the Blavatnik Family
    foundation, the Blavatnik Interdisciplinary Cyber Research Center, Tel Aviv
    University, and the Pazy Foundation.
    This material is based upon work supported by the United States-Israel
    Binational Science Foundation (BSF) grants No. 2016260 and 2012259.
    This research was also supported in part by NSF grants 1704941 and 1420064, and
    funding provided by Intel Corporation.
\end{acknowledgements}

\bibliographystyle{spmpsci}      

\bibliography{refs}

\end{document}